\documentclass[12pt, reqno, oneside]{amsart}

\usepackage{enumerate, microtype, mlmodern, amssymb, tikz, mathtools, ragged2e,addlines}
\usepackage[longnamesfirst]{natbib}
\shortcites{dedeckeretal2007}
\defcitealias{ibragimovmueller2016}{Ibragimov-M\"uller}
\defcitealias{canayetal2014}{Canay-Romano-Shaikh}
\defcitealias{besteretal2014}{Bester-Conley-Hansen}
\defcitealias{conleytaber2011}{Conley-Taber}
\defcitealias{wordetal1990}{Word et al.\ (1990)}

\input{macros}

\theoremstyle{plain}
\newtheorem{theorem}{Theorem}[section]

\newtheorem{proposition}[theorem]{Proposition}
\newtheorem{assumption}[theorem]{Assumption}

\theoremstyle{definition}
\newtheorem{algorithm}[theorem]{Algorithm}

\newtheorem{example}[theorem]{Example}
\theoremstyle{remark}
\newtheorem*{remarks}{Remarks}

\newcommand{\ev}{{\mathord \mathrm{E}}}%
\newcommand{\prob}{{\mathord P}}%
\DeclareMathOperator{\id}{id}%
\DeclareMathOperator*{\var}{Var}%
\newcommand{\pto}{\mathchoice
	{\raisebox{.0em}{ $\overset{\mathrm{P}}{\to}$ }}
	{\raisebox{-.15em}{ $\overset{\raisebox{-.25em}{\scriptsize$\mathrm{P}$}}{\to}$ }}
	{}
	{}}
\newcommand{\wto}[1]{\mathchoice
	{\raisebox{.0em}{ $\overset{#1}{\leadsto}$ }}
	{\raisebox{-.15em}{ $\overset{\raisebox{-.25em}{\scriptsize$#1$}}{\leadsto}$ }}
	{}
	{}}
\newcommand{\xqed}[1]{%
  \leavevmode\unskip\penalty9999 \hbox{}\nobreak\hfill
  \quad\hbox{\ensuremath{#1}}}
\newcommand{\sqed}{\xqed{\square}}

\newcommand{\G}{\mathcal{G}}
\newcommand\mydots{\hbox to 1em{.\hss.\hss.}}

\usepackage{xcolor}
\definecolor{umorange}{HTML}{cc6600}
\definecolor{umblue}{HTML}{587abc}
\definecolor{umgrey}{HTML}{989c97}
\definecolor{umred}{HTML}{7a121c}
\definecolor{umgreen}{HTML}{83b2a8}

\usepackage[bookmarks=false,pdfpagemode=UseNone,pdftex,colorlinks,citecolor=black,%
            filecolor=black,linkcolor=black,urlcolor=black,%
            pdfauthor=Andreas~Hagemann,pdfstartview=FitH]{hyperref}

\begin{document}

\title[Quantile processes with a finite number of clusters]{Inference on quantile processes\\ with a finite number of clusters}
\author[]{Andreas Hagemann}
\address{University of Michigan Ross School of Business, 701 Tappan Ave, Ann Arbor, MI 48109, USA. Tel.: +1 (734) 764-2355. Fax: +1 (734) 764-2769. E-mail: \href{mailto:hagem@umich.edu}{\texttt{hagem@umich.edu}}.}

\date{\today. University of Michigan Stephen M.\ Ross School of Business, 701 Tappan Ave, Ann Arbor, MI 48109, USA. Tel.: +1 (734) 764-2355. Fax: +1 (734) 764-2769. E-mail: \href{mailto:hagem@umich.edu}{\texttt{hagem@umich.edu}}.
I would like to thank two anonymous reviewers for helpful comments. All errors are my own.
}

\begin{abstract}
I introduce a generic method for inference on entire quantile and regression quantile processes in the presence of a finite number of large and arbitrarily heterogeneous clusters. The method asymptotically controls size by generating statistics that exhibit enough distributional symmetry such that randomization tests can be applied. The randomization test does not require ex-ante matching of clusters, is free of user-chosen parameters, and performs well at conventional significance levels with as few as five clusters. The method tests standard (non-sharp) hypotheses and can even be asymptotically similar in empirically relevant situations. The main focus of the paper is inference on quantile treatment effects but the method applies more broadly. Numerical and empirical examples are provided.
\end{abstract}

\maketitle

\section{Introduction}
Economic data often contain large clusters such as countries, regions, villages, or firms. Units within these clusters can be expected to influence one another or are influenced by the same political, environmental, sociological, or technical shocks. Several analytical and computer-intensive procedures such as the bootstrap are available to account for the presence of data clusters. These procedures generally achieve consistency by letting the number of clusters go to infinity. Numerical evidence by \citet{bertrandetal2004}, \citet{mackinnonwebb2014}, and others in the context of mean regression suggests that this type of asymptotic approximation often causes substantial size distortions when the number of clusters is small or the clusters are heterogenous. True null hypotheses are rejected far too often in both situations. \citet{hagemann2017} shows that this phenomenon is also present in quantile regression. 

In this paper, I develop a generic method for inference on the entire quantile or regression quantile process in the presence of a finite number of large and arbitrarily heterogeneous clusters. The method, which I refer to as \emph{cluster-randomized Kolmogorov-Smirnov} (\emph{CRK}) test, asymptotically controls size by generating Kolmogorov-Smirnov statistics that exhibit enough distributional symmetry at the cluster level such that randomization tests \citep{fisher1935, canayetal2014} can be applied. The CRK test is not limited to the pure quantile regression setting and can be used in distributional difference-in-differences estimation \citep{callawayli2019} and related situations where quantile treatment effects are identified by between-cluster comparisons. The CRK test is free of user-chosen parameters, powerful against fixed and root-$n$ local alternatives, and performs well at conventional significance levels with as few as twelve clusters if parameters are identified between clusters. If parameters are identified within clusters, then even five clusters are sufficient for inference. 

Quantile regression (QR), introduced by \citet{koenkerbassett1978}, is an important empirical tool because it can quantify the effect of a set of covariates on the entire conditional outcome distribution. An issue with QR in the presence of clustering is that estimates normalized by their asymptotic covariance kernel have standard normal marginal limit distributions but are no longer pivotal for any choice of weight matrix \citep{hagemann2017}. Cluster-robust tests about the QR coefficient function therefore have asymptotic distributions that cannot be tabulated for inference about 
ranges of quantiles. Even if only individual quantiles are of interest, consistent covariance matrix estimation in large clusters is challenging. It requires knowledge of an explicit ordering of the dependence structure within each cluster combined with a kernel and bandwidth choice to give distant observations less weight. Because time has a natural order, this weighting is easily done for time-dependent data but ordering data within states or villages may be difficult or impossible. The common empirical strategy of simply assuming that the clusters are small and numerous enough to satisfy a central limit theorem circumvents these issues but can lead to substantial size distortions with as few as 20 clusters \citep{hagemann2017}. This remains true if a cluster-robust version of the bootstrap is used. Distortions can be especially severe if clusters differ greatly in their size and dependence structure.

I show that the CRK test is robust to each of these concerns: It performs well even when the number of clusters is small, the dependence varies from cluster to cluster, and the cluster sizes are heterogenous. The reason for this robustness is that the CRK test does not rely on clustered covariance matrices to rescale the estimates. I instead use randomization inference to generate random critical values that automatically scale to the data. There are no kernels, bandwidths, or spatio-temporal orderings of the data to choose. The test achieves consistency with a finite number of large but heterogeneous clusters under interpretable high-level conditions. Despite being based on randomization inference, the CRK test can perform standard (non-sharp) inference on entire quantile or regression quantile processes. Randomization is performed with a fixed set of estimates and does not require repeated estimation to obtain its critical values.

The randomization method underlying the CRK test was first used in the cluster context by \citet{canayetal2014} as a way to perform inference on a finite-dimensional parameter with Student $t$ and Wald statistics in least squares regression. They do not consider inference on quantile functions or Kolmogorov-Smirnov statistics. Here, I considerably extend the scope of their method under explicit regularity conditions to allow for inference on the entire QR process and related objects. The proofs below are fundamentally different from those of \citeauthor{canayetal2014}\ to account for the infinite-dimensional setting and do not rely on the Skorokhod almost-sure representation theorem. A practical issue with their method is that they require treated clusters to be matched ex-ante with an equal number of control clusters. Each match corresponds to a separate test and two researchers working with the same data can reach different conclusions based on which matches they choose. If there is not an equal number of treated and control clusters, then some clusters have to be combined or dropped in an ad-hoc manner. The CRK test sidesteps these issues completely and explicitly merges all potential tests into a single, uniquely determined test decision using results of \citet{ruschendorf1982}.

Cluster-robust inference in linear regression models has a long history; recent surveys include \citet{cameronmiller2014} and \citet{mackinnonnielsenwebb2022}. 
\citet{chenetal2003}, \citet{wanghe2007}, \citet{wang2009}, \citet{parentesantossilva2013}, and \citet{hagemann2017} provide bootstrap and analytical methods for cluster-robust inference in QR models. \citet{yoongalvao2020} discuss the situation where clusters arise from correlation of individual units over time. All of these papers require the number of clusters to go to infinity for consistency. The CRK test differs from these papers because it is based on randomization inference and is consistent with a finite number of clusters. 

Several papers show that pointwise inference with a fixed number of clusters is possible under a variety of conditions. \citet{ibragimovmueller2010, ibragimovmueller2016} use special properties of the Student $t$ statistic to perform inference on scale mixtures of normal random variables. \citet{besteretal2014} use standard cluster-robust covariance matrix estimators but adjust critical values under homogeneity assumptions on the clusters. \citet{canayetal2018} show that certain cluster-robust versions of the wild bootstrap can be valid under strong homogeneity assumptions with a fixed number of clusters. \citet{hagemann2019b} adjusts permutation inference for arbitrary heterogeneity at the cluster level but his bounds only apply to finite-dimensional objects. All of these methods can be used for inference at a single quantile but are not designed for simultaneous inference across ranges of quantiles. In contrast, the CRK test provides uniformly valid inference on the entire quantile process even if clusters are arbitrarily heterogeneous.

The remainder of the paper is organized as follows: Section \ref{s:sym} establishes new results on randomization inference on Gaussian processes. Section \ref{s:asym} uses these results to show consistency of the CRK test and gives specific examples where the test applies, including quantile difference-in-differences. Section~\ref{s:mc} illustrates the finite sample behavior of the test in Monte Carlo experiments and an empirical example using Project STAR data. The appendix contains proofs. 

I use the following notation and definitions: $1\{\cdot\}$ is the indicator function, cardinality of a set $A$ is $|A|$, the smallest integer greater than or equal $a$ is $\lceil a \rceil$, and the largest integer smaller than or equal $a$ is $\lfloor a \rfloor$. The minimum of $a$ and $b$ is denoted by $a \wedge b$. 
Limits are as $n\to\infty$ unless noted otherwise. Convergence in distribution under the parameter $\delta$ is denoted by $\wto{\delta}$. A stochastic process $\{\xi(t) : t\in\mathcal{T}\}$ indexed by a set $\mathcal{T}$ is a collection of random variables $\xi(t)\colon \Omega \to \mathbb{R}$ defined on the same probability space $(\Omega, \mathcal{F}, P)$. Such a process is Gaussian if and only if $(\xi(t_1), \dots, \xi(t_m))$ is multivariate normal for any finite collection of indices $t_1,\dots,t_m\in\mathcal{T}$. 

\section{Randomization inference on Gaussian processes}\label{s:sym}
In this section I study the size of randomization tests when the data come from heterogeneous Gaussian processes. I then analyze asymptotic size when a limiting experiment is characterized by such processes. The next section applies these generic results to the quantile setting.

I first introduce some notation for randomization tests that I will use throughout the paper. Let $u\mapsto X_j(u)$, $1\leq j\leq q$, be independent mean-zero Gaussian processes indexed by $u\in\mathcal{U}$, where $\mathcal{U}$ is a compact subset of $(0,1)$. 
Symmetry about zero implies that $(X_j(u_1), \dots, X_j(u_m))$ and  $-(X_j(u_1), \dots, X_j(u_m))$ are identically distributed. Because this is true for every finite collection of indices $u_1,\dots,u_m\in\mathcal{U}$, $u\mapsto X_j(u)$ and $u\mapsto -X_j(u)$ have the same (finite-dimensional) distributions. Define $\G  = \{1, -1\}^q$ as the $q$-dimensional product of $\{1, -1\}$ and, for $g = (g_1,\dots,g_q)\in \G$, define $g\mapsto gx$ as the direct product $gx = (g_1x_1,\dots, g_qx_q)$ of $g$ and $x\in\mathbb{R}^q$.  Independence and symmetry together imply that $u\mapsto X(u) = (X_1, \dots, X_q)(u)$ and $u\mapsto gX(u)$ have the same distribution for every $g\in\G$ as long as $X$ has mean zero. The quantile and quantile-like processes discussed in the next section have this property under the null hypothesis. Deviations from the null cause non-zero means and therefore also asymmetry in $X$. The goal of this section is to develop a test of the null hypothesis of symmetry about zero, 
\begin{equation}\label{eq:invar}
	H_0\colon X(u) \sim gX(u), \qquad \text{all $g\in\G$, all $u\in\mathcal{U}$}.
\end{equation}

To test this hypothesis, I use the Kolmogorov-Smirnov-type statistic \begin{equation}\label{eq:crk} T(X) = \sup_{u \in \mathcal{U}}\Biggl(\frac{1}{q}\sum_{j=1}^q X_j(u)\Biggr).
\end{equation} 
This statistic is large if symmetry is violated because the mean of the $X_j(u)$ is positive. I focus on one-sided tests to the right for simplicity but this is not restrictive. To test whether the mean is negative, simply use $-X$ instead of $X$ in the definition of $T$. These test statistics can be combined for two-sided tests. I explain this in detail at the end of Section \ref{s:asym}.

Randomization inference uses distributional invariance to generate null distributions and critical values. In the present case, $X$ is distributionally invariant to all transformations $g$ contained in $\G$ because $X$ is symmetric.  Let $T^{(1)}(X, \G ) \leq $ $T^{(2)}(X, \G ) \leq \dots \leq T^{(|\G|)}(X, \G )$ be the $|\G| = 2^q$ ordered values of $T(g X)$ across $g\in\G $ and let \begin{equation}\label{eq:critval}T^{1-\alpha}(X,\G ) := T^{(\lceil(1-\alpha) |\G |\rceil)}(X, \G )\end{equation}
be the $1-\alpha$ quantile of these values. The randomization test function is then 
\begin{equation}\label{eq:crktest}
	\varphi_{\alpha}(X, \G ) = 1\{T(X) > T^{1-\alpha}(X,\G )\}.
\end{equation}

If $\mathcal{U}$ is a finite set, distributional invariance under $H_0$ immediately implies $\ev \varphi_{\alpha}(X, \G ) = \ev \varphi_{\alpha}(gX, \G ).$ By an argument due to \citet{hoeffding1952}, the test function must satisfy $|\G |\alpha \geq \sum_{g\in\G } \varphi_{\alpha}(gX, \G )$ and, after taking expectations on both sides, equality of the distributions yields $|\G |\alpha \geq \ev \sum_{g\in\G } \varphi_{\alpha}(gX, \G ) = \sum_{g\in\G } \ev \varphi_{\alpha}(gX, \G ) =  |\G |\ev \varphi_{\alpha}(X, \G )$. This implies $\ev \varphi_{\alpha}(X, \G ) \leq \alpha$, which makes $T^{1-\alpha}(X,\G )$ an $\alpha$-level critical value. 

If $\mathcal{U}$ is a not finite, this argument does not immediately go through because \eqref{eq:crk} is a statement about possibly uncountably many $u \in \mathcal{U}$ but I have only established equivalence of the finite-dimensional distributions. However, as the following theorem shows, the conclusion that the test controls size holds nonetheless. The proof of the theorem extends \citeauthor{hoeffding1952}'s proof to stochastic processes with smooth sample paths by showing that \eqref{eq:invar} implies equality of the distributions of $(T(gX))_{g\in\G}$ and $(T(g\tilde{g}X))_{g\in\G}$ for every $\tilde{g}\in\G$. I prove that this is enough for \citeauthor{hoeffding1952}'s argument to go through as long as at least one of the processes has positive variance at every $u$.
\begin{theorem}\label{th:hoeffding}
Let $\{X_1(u)\colon u \in \mathcal{U}\}, \dots, \{X_q(u)\colon u\in \mathcal{U}\}$ be independent mean-zero Gaussian processes with continuous sample paths indexed by the compact set $\mathcal{U}\subset (0,1)$ and let $u\mapsto X(u) := (X_1,\dots, X_q)(u)$. If there is a $j\in\{1,\dots, q \}$ such that $\prob(X_j(u) = 0) = 0$ for all $u\in \mathcal{U}$, then $\ev \varphi_{\alpha}(X, \G )\leq \alpha$.
\end{theorem}

\begin{remarks}
(i)~If desired, the test decision can be randomized to construct an exact test. Take an independent variable $V$ with a uniform distribution on $[0,1]$ and the nonrandomized test function 
\begin{equation} \label{eq:randtest}
\phi_{\alpha}(X, \G) = \begin{cases} 1 &\text{if }T(X) > T^{1-\alpha}(X,\G ),\\ a(X) & \text{if }T(X) = T^{1-\alpha}(X,\G ),\\ 0 &\text{if }T(X) < T^{1-\alpha}(X,\G ), \end{cases} 
\end{equation}
 where \[ a(X) = \frac{|\G|\alpha - |\{g\in\G : T(gX) > T^{1-\alpha}(X,\G )|}{|\{g\in\G : T(gX) = T^{1-\alpha}(X,\G )|}. \]
Using arguments of \citet{hoeffding1952}, I show in the proof of Theorem \ref{th:hoeffding} that the randomized test indeed satisfies $\prob (\phi_{\alpha}(X, \G) \geq V) = \alpha$. 
However, this type of test is uncommon in practice because rejecting the null if $\phi_{\alpha}(X, \G)\geq V$ bases the test decision on a single draw from the uniform distribution. A researcher could therefore draw until a desired conclusion was reached.	

(ii)~Similar arguments arise in the context of conformal prediction \citep{vovketal2005} with exchangeable data. Such arguments do not apply here because $(T(gX))_{g\in\G}$ is generally not exchangeable. 
\sqed
\end{remarks}

If $X$ is only an approximation in the sense that $X_n\leadsto X$ in $\ell^\infty(\mathcal{U})^q$, the space of bounded maps from $\mathcal{U}$ to $\mathbb{R}^q$, then the conclusions of the theorem still hold as long as the non-degeneracy conditions are strengthened. Here and in the following I tacitly assume that a process is indexed by a compact $\mathcal{U}\subset (0,1)$ and that $\ell^\infty(\mathcal{U})^q$ is equipped with the Borel $\sigma$-field induced by the uniform norm topology. 
\begin{theorem}\label{th:asysize} If $X_n \leadsto  X = \{ (X_1,\dots, X_q)(u) \colon u\in\mathcal{U}\},$ where the $\{X_j(u) \colon u\in\mathcal{U}\}$ are independent mean-zero Gaussian processes with continuous sample paths that satisfy $\prob(X_j(u) = -X_j(u'))=0$ for all $u,u'\in \mathcal{U}$ and $1\leq j\leq q$, then  $\ev \varphi_{\alpha}(X_n, \G ) \to \ev \varphi_{\alpha}(X , \G )$.
\end{theorem}
\begin{remarks}
(i)~For the non-degeneracy assumption $\prob(X_j(u) = -X_j(u'))=0$ to fail, a Gaussian process with uniformly continuous sample paths has to traverse, with certainty, from $X_j(u)$ to $X_j(u') = -X_j(u)$ while maintaining a positive variance along the entire path. The process would have to have identical variances at time $u$ and $u'$ but be perfectly negatively correlated at those times, which is impossible for Brownian bridges and related processes that typically arise in a quantile context. Still, such Gaussian processes exist and have to be ruled out.	

(ii)~The main difficulty of the proof of Theorem \ref{th:asysize} is that the critical value $T^{1-\alpha}(X_n, \G)$ does not settle down in the limit and is highly dependent on $T(X)$. The assumptions of Theorem~\ref{th:asysize} rule out degeneracies in the limit process that could lead to ties in the order statistics of $\{T(gX) : g\in\G \}$. This would put probability mass on the boundary of the set $\{T(X) > T^{1-\alpha}(X, \G)\}$ and prevent application of the portmanteau lemma. \citet{canayetal2014} use a delicate construction based on Skorokhod's representation theorem to account for the randomness in the limit. While these results could be extended from vectors to processes, I instead give a direct proof that I can also use to analyze the behavior of the test under both local and global alternatives when I discuss quantile processes in the next section.

(iii)~Similar but less involved arguments show that if the supremum in the test statistic \eqref{eq:crktest} is replaced by an integral over $\mathcal{U}$, then Theorems \ref{th:hoeffding} and \ref{th:asysize} continue to hold. However, this implicitly changes \eqref{eq:invar} to an hypothesis about the symmetry of $\int_\mathcal{U}X(u)du$. Other forms of the test statistic can also lead to valid tests, although the smoothness conditions described in parts (i) and (ii) of this remark may change. 
\sqed
\end{remarks}

\section{Inference on quantile processes with a finite number of clusters}\label{s:asym}
This section gives high level conditions under which asymptotically valid inference on quantile processes and related objects can be performed even if the underlying data come from a fixed number of heterogeneous clusters.

\subsection{Inference when parameters are identified within clusters}\label{ss:within}
Suppose data from $q$ large clusters (e.g., counties, regions, schools, firms, or stretches of time) are available. Throughout the paper, the number of clusters $q$ remains fixed and does not grow with the number of observations~$n$. Observations are independent across clusters but dependent within clusters. Data from each cluster $1\leq j\leq q$ separately identify a quantile or quantile-like scalar function $\delta : \mathcal{U}\to \mathbb{R}$. 
The $\delta$ can be estimated by $\hat{\delta}_j$ using data from only cluster~$j$ such that a total of $q$ separate estimates $(\hat{\delta}_1, \dots, \hat{\delta}_q) =: \hat{\delta}$ of $u \mapsto \delta(u)$ are available. The goal is to use randomization inference on a centered and scaled version of $\hat{\delta}$ to develop tests of the null hypothesis 
\begin{equation}\label{eq:h0}
H_0\colon \delta(u) = \delta_0(u),  \quad\text{all } u\in\mathcal{U},	
\end{equation}
for some known function $\delta_0 : \mathcal{U} \to \mathbb{R}$. The following two examples describe simple but empirically relevant situations that fit this framework.

\begin{example}[Regression quantiles]\label{ex:qr}
	Suppose an outcome $Y_{i,j}$ of individual $i$ in cluster $j$ can be represented as $Y_{i,j} = X_{i,j}\delta(U_{i,j}) + Z_{i,j}'\beta_j(U_{i,j})$, where $u\mapsto X_{i,j}\delta(u) + Z_{i,j}'\beta_j(u)$ is strictly increasing in $u$ and $U_{i,j}$ is standard uniform conditional on covariates $(X_{i,j}, Z_{i,j})$. Here $X_{i,j}$ is the scalar covariate of interest and the $Z_{i,j}$ are additional controls. Monotonicity implies that the $u$-th conditional quantile of $Y_{i,j}$ is $X_{i,j}\delta(u) + Z_{i,j}'\beta_j(u)$ and linear QR as in \citet{koenkerbassett1978} can provide estimates $(\hat{\delta}_j,\hat{\beta}_j)$ of $(\delta, \beta_j)$ for each cluster. Testing \eqref{eq:h0} with $\delta_0 \equiv 0$ tests whether $Y_{i,j}$ and $X_{i,j}$ are associated at any quantile after controlling for $Z_{i,j}$. 
	
Several related models fit the framework of this example: (i)~The $\beta_j$ can be constant across clusters. This does not impact the null hypothesis or the computation of the $\hat{\delta}_j$. (ii)~The $\delta$ can vary by cluster in the QR model $Y_{i,j} = X_{i,j}\delta_j(U_{i,j}) + Z_{i,j}'\beta(U_{i,j})$ under the alternative. This has no impact on the computation of the ${\delta}_j$ and the null hypothesis simply becomes $H_0 \colon \delta_1 = \dots = \delta_q = \delta_0$. Identical $\delta_j$ are required only under the null hypothesis. (iii)~If $\beta_j\equiv 0$ and $X_{i,j}\equiv 1$, then $u \mapsto \hat{\delta}(u)$ reduces to the $u$-th unconditional empirical quantile of $Y_{i,j}$. The null \eqref{eq:h0} can then be used to test whether $\delta$ has a specific functional form, e.g., a standard normal quantile function. \sqed
\end{example}

\begin{example}[Quantile treatment effects]\label{ex:qte} Consider predetermined pairs $\{(j, j+q) : 1\leq j\leq q \}$  of $2q$ groups. Suppose the first $q$ groups received treatment, indicated by $D_{j} = 1\{j\leq q\}$, and the remaining groups did not. Groups here could be manufacturing plants or villages. Treatment could be management consulting or introduction of a new technology. Denote treatment and control potential outcomes by $Y_{j}(1)\sim F_{Y(1)}$ and $Y_{j}(0)\sim F_{Y(0)}$, respectively. The observed outcome is $Y_{j} = D_{j}Y_{j}(1) + (1-D_{j})Y_{j}(0)$. For each group $j$, the experimenter observes identically distributed but potentially highly dependent copies $Y_{i,j}$ of $Y_{j}$ representing workers $i$ within group $j$. View each pair $(j, j+q)$ for $1\leq j\leq q$ as a cluster and define the quantile treatment effect (QTE) as\[ u\mapsto \delta(u) = F^{-1}_{Y(1)}(u) - F^{-1}_{Y(0)}(u). \] This QTE can be estimated as difference of the empirical quantiles \[ u\mapsto \hat{\delta}_j(u) = \hat{F}^{-1}_{Y_{j}}(u) - \hat{F}^{-1}_{Y_{j+q}}(u) \] or, alternatively, as the coefficient on $D_j$ in a QR of $Y_{i,j}$ on a constant and $D_{j}$ using data only from cluster $j$. The situation where $\delta$ varies with $j$ is again included in the analysis as long as the null hypotheses is $\delta_1 = \dots = \delta_q = \delta_0$. Estimation remains unchanged. I discuss the more complex scenario where the counterfactual $F_{Y(0)}$ has to be identified through difference-in-differences methods in Example \ref{ex:qdid} ahead.
	\sqed
\end{example}

The $\hat{\delta}$ is neither limited to the estimators discussed in the preceding two examples nor does it need to have a special functional form. However, I assume that it can be approximated by a Gaussian process as in Theorem~\ref{th:asysize}. Let $1_q$ be a $q$-vector of ones.
\begin{assumption}\label{as:weakc}
The stochastic process $\{ \hat{\delta}(u) \colon u \in \mathcal{U}\}$ with $\hat{\delta}(u)\in\mathbb{R}^q$ satisfies
\begin{equation}\label{eq:wc} 
X_n := \{\sqrt{n}(\hat{\delta}-\delta 1_q)(u) : u\in \mathcal{U}\} \wto{\delta} X = \{ (X_1,\dots, X_q)(u) \colon u\in\mathcal{U} \},
\end{equation}  
where the components of $X$ are independent mean-zero Gaussian processes with continuous sample paths, $\prob(X_j(u) = -X_j(u'))=0$ for all $u,u'\in \mathcal{U}$ and $1\leq j\leq q$.
\end{assumption}
\noindent Examples of $X_n$ that can satisfy this assumption include unconditional quantile functions, coefficient functions in quantile regressions, quantile treatment effects, and other quantile-like objects. \citet{machkouriaetal2013} present invariance principles and moment bounds that can be used to establish the convergence condition \eqref{eq:wc} under explicit weak dependence conditions. 

I now connect the results from Section \ref{s:sym} about heterogeneous Gaussian processes to tests about $\delta$ under Assumption \ref{as:weakc}. The key property is that if $H_0$ in \eqref{eq:h0} does not hold, then $\sqrt{n}(\hat{\delta}-\delta_0 1_q) = X_n + \sqrt{n}(\delta - \delta_0)1_q$. The $X_n$ converges to a symmetric process but $\sqrt{n}(\delta-\delta_0)(u)$ grows without bound for some $u$, which makes the distribution of $\sqrt{n}(\hat{\delta}-\delta_0 1_q)$ highly asymmetric. Testing for symmetry using randomization inference is therefore informative about the hypothesis that $\delta=\delta_0$. I refer to a test that uses $\hat{\delta} - \delta_0 1_q$ in place of $X$ in test function \eqref{eq:crktest} as the \emph{cluster-randomized Kolmorogov-Smirnov} (\emph{CRK}) test. From a practical perspective, the function $\delta_0$ is almost always $\delta_0 \equiv 0$. This tests the null of no effect at any quantile but more general hypotheses can be considered.  

The test function $x\mapsto\varphi_{\alpha}(x, \G )$ is invariant to scaling of $x$ by positive constants. If $H_0\colon \delta = \delta_0$ is true, then the CRK test satisfies \[T(\hat{\delta}-\delta_0 1_q) > T^{1-\alpha}(\hat{\delta}-\delta_0 1_q,\G )\] if and only if $T(X_n) > T^{1-\alpha}(X_n,\G )$. That the CRK test is an asymptotic $\alpha$-level test 
is then an immediate consequence of Theorems~\ref{th:hoeffding} and \ref{th:asysize}.
\begin{theorem}[Size]\label{th:size}
Suppose Assumption \ref{as:weakc} holds. If $H_0\colon \delta = \delta_0$ is true, then $\lim_{n\to\infty} \ev \varphi_{\alpha}(\hat{\delta} - \delta_0 1_q, \G ) \leq \alpha.$
\end{theorem}
\begin{remarks}
(i)~The canonical limit of quantile and regression quantile processes such as those in Examples~\ref{ex:qr} and \ref{ex:qte} is a scaled version of a $q$-dimensional Brownian bridge. That process easily satisfies the non-standard condition $\prob(X_j(u) = -X_j(u'))=0$ imposed by Assumption \ref{as:weakc}.

(ii)~The inequality in the theorem becomes an equality if $(1-\alpha)2^q$ is an integer. In that case, the test in the limit experiment is ``similar,'' i.e., it has rejection probability exactly equal to $\alpha$ for all Gaussian processes that satisfy Assumption \ref{as:weakc}. The CRK test can therefore be asymptotically similar in some situations. If desired, the test decision can be randomized to make the CRK test similar in the limit for all $\alpha$.
\sqed
\end{remarks}

To analyze the power of the CRK test, I consider fixed alternatives $\delta(u) = \delta_0(u) + \lambda(u)$ with a positive function $u\mapsto \lambda(u)$, and local alternatives $\delta(u) = \delta_0(u) + \lambda(u)/\sqrt{n}$ converging to the maintained null hypothesis $H_0\colon \delta = \delta_0$.  In the local case, $\delta_0$ is fixed but $\delta$ now depends on $n$ and the convergence \eqref{eq:wc} is under the sequence of functions $\delta = \delta_0 + \lambda/\sqrt{n}$. As the following results show, the CRK test has power against both types of alternatives.
\begin{theorem}[Global and local power]\label{th:power}
	Suppose Assumption \ref{as:weakc} holds and $\alpha \geq 1/2^{q} $. If $H_1\colon \delta = \delta_0 + \lambda$ is true with $\lambda \colon \mathcal{U} \to [0,\infty)$ continuous and $\sup_{u\in\mathcal{U}} \lambda(u) > 0$, then $\lim_{n\to\infty}\ev \varphi_{\alpha}(\hat{\delta}-\delta_0 1_q, \G ) = 1$. If $H_1\colon \delta = \delta_0 + \lambda/\sqrt{n}$ is true with $\sup_{u\in\mathcal{U}} \lambda(u) > \ev \sup_{u\in\mathcal{U}} X_j(u)$, $1\leq j\leq q$, then \[ \lim_{n\to\infty }\ev \varphi_{\alpha}(\hat{\delta}-\delta_0 1_q, \G )\geq \prod_{j=1}^q \Bigl( 1 - e^{-[\sup\lambda(u) - \ev \sup X_j(u)]^2/2 \sup \ev X^2_j(u)}\Bigr) > 0, \] where the suprema in the exponent are over $u\in\mathcal{U}$.
\end{theorem}
\begin{remarks}
	(i)~The lower bound used for the local power result comes from the Borell-Tsirelson-Ibragimov-Sudakov (Borell-TIS) inequality \citep[see, e.g.,][p.\ 50]{adlertaylor2007}. For large $q$, the bound is relatively crude but for small $q$, the only crude part is the assumption that $\delta$ is moderately large when compared to $X$. This is reflected in the condition that $\sup_{u\in\mathcal{U}} \lambda(u) > \ev \sup_{u\in\mathcal{U}} X_j(u)$ instead of $\sup_{u\in\mathcal{U}} \lambda(u) > 0$. The bound can be made arbitrarily close to $1$ by choosing $\sup_{u\in\mathcal{U}} \lambda(u)$ large enough.
	
	(ii)~ If $(1-\alpha)|\G | > |\G |-1$, the power of the test is identically zero. In that case $T^{1-\alpha}(X,\G) = \max_{g\in\G}T(gX)$ and $T(X) > T^{1-\alpha}(X,\G)$ becomes impossible because $T(X)$ is contained in $\{T(gX) : g\in\G \}$. I therefore I focus on the case $(1-\alpha)|\G |\leq |\G |-1$, which is equivalent to $\alpha \geq 1/2^q$. 

	(iii)~The test also has power against alternatives where $\lambda$ varies with the cluster index $j$ and at least some of the $\lambda_j$ are large. However, a precise statement without additional conditions on the relative sizes of the $\lambda_j$ is involved. I do not pursue this here to prevent notational clutter.
	\sqed
\end{remarks}

\subsection{Inference when parameters are identified across clusters} \label{ss:between}
In applications, the treatment effect is often not identified from within a cluster but by comparisons across two clusters. This is the case, for example, if treatment is assigned at random at the cluster level or if identification comes from comparing changes in one cluster to changes in another cluster in a quasi-experimental context. In this situation, each individual pairing of a treated cluster $j$ with a control cluster $k$ is generally informative about the treatment effect of interest $\delta$ and each $(j,k)$ pair gives rise to an estimate $\hat{\delta}_{j,k}$ of $\delta$ that could be used in a CRK-type test. The following example illustrates this for difference-in-differences estimation of quantile treatment effects.
\begin{example}[Quantile difference in differences]\label{ex:qdid} Let $\Delta Y_t(0) = Y_t(0) - Y_{t-1}(0)$ be time differences of untreated outcomes. Periods $t \in \{0,-1 \}$ are pre-intervention periods and $t=1$ is the post-intervention period; $Y_1(1)$ is a treated potential outcome and $Y_t$ are observed outcomes. 
Denote by $F_{Y\mid D=d}$ the distribution of a variable $Y$ conditional on the treatment indicator taking on the value $d\in\{0,1\}$. \citet{callawayli2019} show that the distribution $F_{Y_1(0)\mid D =1 }(y)$ of the untreated potential outcome of a treated observation at time $t=1$ can be identified as 
\begin{equation}\label{eq:qttcounterfactual}
	 \prob\Bigl( F^{-1}_{\Delta Y_1\mid D=0} \bigl(F_{\Delta Y_{0} \mid D = 1}(\Delta Y_{0})\bigr) + F^{-1}_{Y_{0}\mid D=0} \bigl(F_{Y_{-1} \mid D = 1}(Y_{-1})\bigr) \leq y\mid D = 1\Bigr) 
\end{equation} as long as a distributional version of the standard parallel trends assumption and some additional stability and smoothness conditions hold. This identifies the quantile treatment on the treated (QTT) effect \[u\mapsto \delta(u) =  F^{-1}_{Y_1(1)\mid D=1}(u) - F^{-1}_{Y_1(0)\mid D=1}(u),\] where $F^{-1}_{Y_1(1)\mid D=1}(u)$ can be estimated by the sample quantile $\hat{F}^{-1}_{Y_1\mid D=1}(u)$. To estimate the counterfactual quantile, \citeauthor{callawayli2019} 
replace $\prob$ and every $F$ in \eqref{eq:qttcounterfactual} with sample equivalents. This yields the estimated QTT 
\begin{equation}\label{eq:qteest}
u\mapsto \hat{F}^{-1}_{Y_1\mid D=1}(u) - \hat{F}^{-1}_{Y_1(0)\mid D=1}(u).	
\end{equation}
\citeauthor{callawayli2019} show that $\sqrt{n}(\hat{F}^{-1}_{Y_1\mid D=1} - \hat{F}^{-1}_{Y_1(0)\mid D=1} - \delta)$ converges to a well-behaved Gaussian process under mild regularity conditions.

Suppose that data come from $q_1$ states that received treatment and $q_0$ states that did not. View a single state over time as a cluster. Then two clusters are enough to compute \eqref{eq:qteest}: $\hat{F}^{-1}_{Y_1\mid D=1}$ can be computed from a treated cluster $j$ and $\hat{F}^{-1}_{Y_1(0)\mid D=1}$ can be computed from $j$ and an untreated cluster $k$. Denote by $\hat{\delta}_{j,k}$ the QTT estimated in this fashion using only data from clusters $j$ and $k$. Each $(j,k)$ pair provides a valid estimate of $\delta$ and each $\hat{\delta}_{j,k}$ could potentially be used in a CRK-type test of the null hypothesis $H_0\colon \delta = \delta_0$.
 \sqed
\end{example}

I again assume that centered and scaled $\hat{\delta}_{j,k}$ converge in distribution to non-degenerate Gaussian processes with smooth sample paths as in Assumption~\ref{as:weakc}. I only adjust this condition for the fact that estimates are constructed from pairwise combination of clusters. Let $q_1$ be the number of treated clusters and let $q_0$ be the number of control clusters.
\begin{assumption}\label{as:weakc2} The process $\{\sqrt{n}(\hat{\delta}_{j,k} - \delta)(u) : u\in \mathcal{U}\}$ converges, jointly in $j$ and $k$, in distribution to mean-zero Gaussian processes $X_{j,k}$ with continuous sample paths that satisfy $\prob(X_{j,k}(u) = -X_{j,k}(u'))=0$ for all $u,u'\in \mathcal{U}$, $1\leq j\leq q_1$, and $1\leq k\leq q_0$. If both $j\neq j'$ and $k\neq k'$, then $X_{j,k}$ and $X_{j',k'}$ are independent.
\end{assumption}

A na\"ive test of $H_0\colon \delta \equiv \delta_0$ would now take $X_{n,j,k} := \sqrt{n}(\hat{\delta}_{j,k} - \delta_{0})$ and generate randomization distributions from $\{X_{n,j,k} : 1\leq j\leq q_1, 1\leq k\leq q_0\}$ via sign changes. However, $X_{n,j,k}$ and $X_{n,j,k'}$ are dependent for any choice of $j,k,k'$ because $j$ is used twice. This remains true even in large samples and if the data from all $q_1+q_0$ groups are independent. Dependence causes problems because $(X_{n,j,k}, X_{n,j,k'})$ and $(X_{n,j,k}, -X_{n,j,k'})$ generally do not have the same joint distribution even when $n\to\infty$. Invariance under transformations with $g$ therefore fails. This issue can be avoided if one works with a subset of $\{X_{n,j,k} : 1\leq j\leq q_1, 1\leq k\leq q_0\}$ that uses each $j$ and $k$ only once. While this solves the dependence issue, it introduces another problem: each of the $q_1$ treatment groups now has to be paired with exactly one of the $q_0$ control groups. Unless these pairings are determined before the data are analyzed, two researchers working with the same data and methodology could arrive at different conclusions because they chose different pairings. To address this problem, I now develop a method that maintains invariance under sign changes but avoids any decisions on the part of the researcher.

I first introduce some notation. If $q_1\leq q_0$, there are $q_0 \times (q_0-1) \times \dots \times (q_0-q_1+1)$ ways of choosing $q_1$ ordered elements out of $(1,\dots, q_0)$. Identify each such choice with an $h$ and denote the collection of all $h$ by $\mathcal{H}$. The ordering within $\mathcal{H}$ will not affect the test decision. For each $h\in\mathcal{H}$, denote by
\begin{equation} \label{eq:hq1leqq0}
	\hat{\delta}_{[h]} = (\hat{\delta}_{1,h(1)}, \hat{\delta}_{2,h(2)}, \dots, \hat{\delta}_{q_1,h(q_1)}), \qquad q_1\leq q_0,
\end{equation}
the vector that matches the subset of control groups associated with the label $h = (h(1),\dots, h(q_1))$ to the (unpermuted) treated groups. If there are more treated than control groups such that $q_1 > q_0$, permute treated groups instead and take $h$ as enumerating ways of choosing $q_0$ elements out of $(1,\dots, q_1)$ to define 
\begin{equation} \label{eq:hq1geqq0}
\hat{\delta}_{[h]} = (\hat{\delta}_{h(1), 1}, \hat{\delta}_{h(2), 2}, \dots, \hat{\delta}_{h(q_0), q_0}), \qquad q_1 > q_0.
\end{equation}
By construction, the entries of $\hat{\delta}_{[h]}$ are independent of one another but $\hat{\delta}_{[h]}$ and $\hat{\delta}_{[h']}$ for $h,h'\in\mathcal{H}$ are potentially highly dependent.

To address the issue that there are multiple ways of combining clusters, I use an adjustment based on the randomization \emph{p}-value 
\begin{equation}\label{eq:pval}
p(X, \G) =  \inf\{ p \in (0,1) : T(X) >  T^{p}(X,\G) \} = \frac{1}{|\G|}\sum_{g\in\G}1\{ T(gX) \geq T(X) \}.
\end{equation}
Testing with this $p$-value is equivalent to a test with a critical value because $T(X) >  T^{1-\alpha}(X,\G)$ if and only if $p(X, \G) \leq \alpha$. The multiple comparisons adjustment is based on an inequality of \citet{ruschendorf1982}. It states that arbitrary, possibly dependent variables $U_h$ indexed by $h\in\mathcal{H}$ with the property that $\prob(U_h\leq u)\leq u$ for every $u\in [0,1]$ satisfy
\begin{equation}\label{eq:rueschendorf}
	\prob\Biggl(\frac{2}{|\mathcal{H}|}\sum_{h\in\mathcal{H}} U_h \leq u\Biggr) \leq u, \qquad \text{every $u\in [0,1]$}.
\end{equation}
This specific form of the inequality is given in \citet{vovkwang2020}. Here the indexing set $\mathcal{H}$ is arbitrary and does not need to be related to permutations. The only condition is that $H = |\mathcal{H}| \geq 2$. The randomization $p$-value $p(\hat{\delta}_{[h]} - \delta_0 1_{q_1  \wedge q_0}, \G)$ for testing whether the treatment effect of interest equals $\delta_0$ can be expected to behave like the $U_h$ in \eqref{eq:rueschendorf} in a large enough sample. Combining $p$-values of the CRK test to reject the null if 
\begin{equation}\label{eq:avgpval}
	\frac{2}{H}\sum_{h\in\mathcal{H}} p(\hat{\delta}_{[h]} - \delta_0 1_{q_1  \wedge q_0}, \G)	
\end{equation}
does not exceed $\alpha$ should then asymptotically control size. The following theorem confirms that this is indeed true. 
\begin{theorem}[Size with combined \emph{p}-values]\label{th:size2}
Suppose Assumption \ref{as:weakc2} holds. If $\delta = \delta_0$, then \[\limsup_{n\to\infty} \prob\Biggl( \frac{2}{H}\sum_{h\in\mathcal{H}} p(\hat{\delta}_{[h]} - \delta_0 1_{q_1  \wedge q_0}, \G) \leq \alpha\Biggr)\leq \alpha	.\]
\end{theorem}
\begin{remarks}
(i)~The theorem can be improved slightly if $\alpha |\G| H/2$ is not an integer. In that case, the limit superior in the theorem is a proper limit that equals $\prob( (2/H)\sum_{h\in\mathcal{H}} p(X_{[h]},\G) \leq \alpha ),$ where $X_{[h]}$ is the weak limit of $\sqrt{n}(\hat{\delta}_{[h]} - \delta_0 1_{q_1  \wedge q_0})$. 
This is because the sum in the preceding display can vary discontinuously at certain values. The limit inferior is $\prob( (2/H)\sum_{h\in\mathcal{H}} p(X_{[h]},\G) < \alpha ).$

(ii)~Results of \citet{vovkwang2020} suggest that other ways of combining $p$-values such as $\exp(1)$ times the geometric mean of the $p$-value instead of a twice the average $p$-value are likely to be applicable here as well. However, the proof of the theorem given here relies crucially on the properties of the \citeauthor{ruschendorf1982} inequality. In the Monte Carlo experiments in the next section, I do not find evidence that other ways of combining $p$-values lead to better results. \sqed
\end{remarks}

The price paid for not matching treated and control clusters before the analysis is lower relative power. When $p$-values are averaged, \citeauthor{ruschendorf1982}'s inequality essentially decreases $\alpha$ to $\alpha/2$ to control size. \citet{meng1993} shows that the constant $2$ cannot be improved. Still, as I establish below, the test has power against global and local alternatives if $\alpha > 1/2^{q_1  \wedge q_0-1}$, which is slightly stronger than what is needed in Theorem \ref{th:power}. Compared to Theorem~\ref{th:power}, I also do not state an explicit bound for the local power analysis because applying the Borel-TIS inequality to the averaged $p$-values directly yields only relatively crude results. I instead show that if the alternatives $\lambda/\sqrt{n}$ converging to the null hypothesis are scaled up by a constant $c$, the test can detect these alternatives in the limit experiment with arbitrary accuracy if $c$ is large enough, that is, if first $n\to\infty$ and then $c\to\infty$.
\begin{theorem}[Global and local power with combined \emph{p}-values]\label{th:power2}
Suppose Assumption~\ref{as:weakc2} holds and $\alpha > 1/2^{q_1  \wedge q_0-1} $. If $H_1\colon \delta = \delta_0 + \lambda$ with $\lambda \colon \mathcal{U} \to [0,\infty)$ continuous and $\sup_{u\in\mathcal{U}}(u) > 0$, then $\lim_{n\to\infty} \prob( (2/H)\sum_{h\in\mathcal{H}} p(\hat{\delta}_{[h]} - \delta_0 1_{q_1  \wedge q_0}, \G) \leq \alpha) = 1$. If $H_1\colon \delta = \delta_0 + c \lambda/\sqrt{n}$, then \[ \lim_{c\to\infty} \liminf_{n\to\infty } \prob\Biggl( \frac{2}{H}\sum_{h\in\mathcal{H}} p(\hat{\delta}_{[h]} - \delta_0 1_{q_1  \wedge q_0}, \G) \leq \alpha\Biggr) = 1. \] 
\end{theorem}

\subsection{Implementation} I now turn to some practical aspects of the CRK test. I discuss (i)~what to do if $\G$ is large, (ii)~what to do if $\mathcal{H}$ is large, and (iii)~how to implement the test with a step-by-step guide. 
First, $\G$ can be prohibitively large if the number of clusters is large. If computing the entire randomization distribution is too costly, then $\G$ can be approximated by a random sample $\G_m$ consisting of $m$ draws from $\G$ with replacement. This is often referred to as ``stochastic approximation.'' The theorems presented in Sections \ref{ss:within} and \ref{ss:between} continue to hold if $\G_m$ is used in place of $\G$ as long as a limit superior or inferior as $m\to \infty$ is applied before $n\to\infty$. The order of limits is not restrictive because, in a given sample of size $n$, the number of draws can $m$ always be made as large as computationally feasible. Under stochastic approximation, the statement in Theorem~\ref{th:size} becomes $\lim_{n\to\infty}\limsup_{m\to\infty} \ev \varphi_{\alpha}(\hat{\delta} - \delta_0 1_q, \G_m ) \leq \alpha$, whereas statements about power use a limit inferior. Limit superior and inferior are needed here because of potential discontinuities but can be replaced by regular limits for most values of $\alpha$.  Theorems \ref{th:size}, \ref{th:size2}, and \ref{th:power2} hold without additional conditions but the conditions of Theorem~\ref{th:power} have to be strengthened marginally to avoid a discontinuity at $\alpha = 1/2^q$. 
\begin{proposition}\label{pr:stochastic}
Suppose $\G_m$ consists of $m$ iid draws from $\G$. If every instance of $\G$ is replaced by $\G_m$, then
\begin{enumerate}[\upshape (i)]
	\item Theorem~\ref{th:size} holds if $\lim_{n\to\infty}$ is replaced by $\lim_{n\to\infty}\limsup_{m\to\infty}$,
	\item Theorem~\ref{th:power} holds if every $\lim_{n\to\infty}$ is replaced by $\lim_{n\to\infty}\liminf_{m\to\infty}$ and $\alpha > 1/2^q$,
	\item Theorem~\ref{th:size2} holds if $\limsup_{n\to\infty}$ is replaced by $\limsup_{n\to\infty}\limsup_{m\to\infty}$,
	\item Theorem~\ref{th:power2} holds if $\lim_{n\to\infty}$ is replaced by $\lim_{n\to\infty}\liminf_{m\to\infty}$ and $\liminf_{n\to\infty}$ is replaced by $\liminf_{n\to\infty}\liminf_{m\to\infty}$.
\end{enumerate}
If $\alpha \not\in \{ j/|G| : 1\leq j\leq |G| \}$, then $\liminf_{m\to\infty}$ and $\limsup_{m\to\infty}$ can be replaced by $\lim_{m\to\infty}$ in {\upshape (i)-(iv)}.
\end{proposition}

Second, the number of elements of $\mathcal{H}$ can similarly be large if the number of clusters is large or if there is a large discrepancy between the number of treated and the number of control clusters. In that case one can again work with a random subset $\mathcal{I}$ of $\mathcal{H}$. The crucial difference to the preceding result is that both Theorems~\ref{th:size2} and \ref{th:power2} continue to hold even if $\mathcal{I}$ consists of only a finite number of random draws. In fact, the result goes through for any $\mathcal{I}$ as long as $\mathcal{I}$ is independent of the data.
\begin{proposition}\label{pr:randomrueschendorf}
	Let $\mathcal{I}$ with $|\mathcal{I}|\geq 2$ be a fixed or random subset of $\mathcal{H}$ independent of the data. Then Theorems \ref{th:size2} and \ref{th:power2} continue to hold if $\mathcal{H}$ is replaced by $\mathcal{I}$.
\end{proposition}

Finally, the following two algorithms outline and summarize how to apply the CRK test in practice. By Theorems \ref{th:size} and \ref{th:size2}, the procedures provide an asymptotically $\alpha$-level test in the presence of a finite number of large clusters that are arbitrarily heterogeneous. They are free of nuisance parameters and do not require any decisions on the part of the researcher. By Theorems \ref{th:power} and \ref{th:power2}, the tests are able to detect fixed and $1/\sqrt{n}$-local alternatives. The first algorithm describes the CRK test when the parameters are identified within clusters. The second algorithm describes the between-cluster case, which is needed for distributional difference in differences. The tests can be two-sided or one-sided in either direction. 

\begin{algorithm}[CRK test for parameters identified within clusters]\label{al:within}\item[]
\begin{enumerate}
	\item Compute for each $j = 1,\dots, q$ and using only data from cluster $j$ an estimate $\hat{\delta}_{j}$ of a parameter of interest $\delta$. (See Examples \ref{ex:qr} and \ref{ex:qte}.) Define $\hat{\delta} = (\hat{\delta}_{1},\dots, \hat{\delta}_{q})$. 
	\item Compute $\G$, the set of all vectors of length $q$ with entries $1$ or $-1$, or replace $\G$ with a large random sample $\G_m$ from $\G$ in the following.
	\item\label{al:practice4} Reject the null hypothesis $H_0\colon \delta(u) = \delta_0(u)$ for all $u$ (e.g., $\delta_0\equiv 0$ tests for no effect of treatment) against
	\begin{enumerate}
		\item $\delta(u) > \delta_0(u)$ for some $u$ if $T(\hat{\delta}-\delta_0 1_q) >  T^{1-\alpha}(\hat{\delta}-\delta_0 1_q,\G)$ for a test with asymptotic level $\alpha$,
		\item $\delta(u) < \delta_0(u)$ for some $u$ if $T(\hat{\delta}-\delta_0 1_q) <  T^{\alpha}(\hat{\delta}-\delta_0 1_q,\G)$ for a test with asymptotic level $\alpha$,
		\item $\delta(u) \neq \delta_0(u)$ for some $u$ if (a) or (b) are true for a test with asymptotic level $2\alpha$,
	\end{enumerate}
	where $T$ is defined in \eqref{eq:crk} and $T^{1-\alpha}(\cdot, \G)$ is the $\lceil (1-\alpha)|\G|\rceil$-th largest value of the randomization distribution of $T$, defined in \eqref{eq:critval}.
\end{enumerate}
\end{algorithm}

\begin{algorithm}[CRK test for parameters identified between clusters]\label{al:between}\item[]
\begin{enumerate}
	\item Compute $\mathcal{H}$, as defined above \eqref{eq:hq1leqq0}, or replace $\mathcal{H}$ with a large subset $\mathcal{I}$ in the following.
	\item Compute $\G$, the set of all vectors of length $q$ with entries $1$ or $-1$, or replace $\G$ with a large random sample $\G_m$ from $\G$ in the following.
	\item For each $h$, compute $\hat{\delta}_{[h]}$ from \eqref{eq:hq1leqq0} if $q_1\leq q_0$ or from \eqref{eq:hq1geqq0} if $q_1 > q_0$. (See Example~\ref{ex:qdid}.) Use \eqref{eq:pval} and \eqref{eq:avgpval} to compute \begin{equation}\label{eq:avgpval2}
\frac{2}{|\mathcal{H}|}\sum_{h\in\mathcal{H}} p(\hat{\delta}_{[h]} - \delta_0 1_{\min\{q_1,q_0\}}, \G) \leq \alpha	
\end{equation}
	\item Reject the null hypothesis $H_0\colon \delta(u) = \delta_0(u)$ for all $u$ (e.g., $\delta_0\equiv 0$ tests for no effect of treatment) against
	\begin{enumerate}
		\item $\delta(u) > \delta_0(u)$ for some $u$ if \eqref{eq:avgpval2} is true for a test with asymptotic level $\alpha$,
		\item $\delta(u) < \delta_0(u)$ for some $u$ if \eqref{eq:avgpval2} is true when $\hat{\delta}_{[h]} - \delta_0 1_{\min\{q_1,q_0\}}$ is replaced by $-(\hat{\delta}_{[h]} - \delta_0 1_{\min\{q_1,q_0\}})$ for a test with asymptotic level $\alpha$,
		\item  $\delta(u) \neq \delta_0(u)$ for some $u$ if (a) or (b) are true for a test with asymptotic level $2\alpha$.
	\end{enumerate}
\end{enumerate}

\end{algorithm}

In some contexts, Algorithm \ref{al:within} can be used even if the parameter of interest is identified by comparisons between treated and untreated clusters. For this to work, the researcher has to merge each treated cluster with an untreated cluster into a single cluster to recover within-cluster identification.  If the number of treated clusters and control clusters is equal, then every treated cluster can be matched with a control cluster according to some rule. If the number of clusters is not equal, then two or more clusters can be merged to force an equal number of treated and control clusters. The merged clusters can then be reinterpreted as clusters and Algorithm~\ref{al:within} can be applied to these new clusters. While this comes with a large number of decisions, it is a valid method for inference if these decisions are made before the data are analyzed. For example, when estimating quantile treatment effects, a pre-analysis plan can be put in place that prescribes how clusters that received treatment will be merged with clusters that did not receive treatment. This reduces the problem to the one described in Example \ref{ex:qte}.

The next section investigates the finite sample performance of Algorithms~\ref{al:within} and \ref{al:between} in several sitations.

\section{Numerical results}\label{s:mc}

This section presents several Monte Carlo experiments to investigate the small-sample properties of the CRK test in comparison to other methods of inference. I discuss significance tests on quantile regression coefficient functions (Example~\ref{ex:qr2}), inference in experiments when parameters are identified between clusters (Example~\ref{ex:qte2}), and estimation of QTEs in Project STAR (Example~\ref{ex:star}). I test one-sided hypotheses to the right but the results apply more broadly.

\begin{example}[Regression quantiles, cont.] \label{ex:qr2} In this example, I adapt an experiment of \citet{hagemann2017} and use the data generating process (DGP) 
\begin{align*}
Y_{i,j,k} = U_{i,j,k} + U_{i,j,k}Z_{i,j,k},
\end{align*}
where $U_{i,j,k} = \sqrt{\varrho}V_{j,k} + \sqrt{1-\varrho}W_{i,j,k}$ with $\varrho \in [0, 1)$;  $V_{j,k}$ and $W_{i,j,k}$ are standard normal, independent of one another, and independent across indices. This ensures that the $U_{i,j,k}$ are standard normal and, for a given $j,k$, any pair $U_{i,j,k}$ and $U_{i',j,k}$ has correlation $\varrho$.  The $Z_{i,j,k}$ satisfy $Z_{i,j,k} = X^2_{i,j,k}/3$ with $X_{i,j,k}$ standard normal independent of $U_{i,j,k}$ to ensure that the $U_{i,j,k} Z_{i,j,k}$ have mean zero and variance one. Both $X_{i,j,k}$ and $U_{i,j,k}$ are independent across $j$ and $k$, and $X_{i,j,k}$ is also independent across $i$. I discard information on $k$ after data generation and drop the $k$ subscripts in the following because they are not assumed to be known. This induces a dependence structure where each cluster $j = 1,\dots, q$ consists of several (unknown) neighborhoods $k = 1,\dots,K$ where observations are dependent if they come from the same $k$ but are independent otherwise. If $K\to\infty$ and the size of the neighborhoods is fixed or grows slowly with $K$, then this dependence structure is compatible with Assumptions~\ref{as:weakc} and \ref{as:weakc2} because it generates the weak dependence needed for central limit theory. In the experiments ahead, I set $K$ to either 10 or 20 and draw the size of each neighborhood from the uniform distribution on $\{5,6,\dots,15\}$.
The DGP in the preceding display corresponds to the QR model
\begin{align}\label{eq:mcqr}
Q(u\mid X_{i,j}, Z_{i,j}) = \beta_0(u) + \beta_1(u)X_{i,j} + \beta_2(u)Z_{i,j}
\end{align}
with $\beta_1(u)\equiv 0$ and $\beta_0(u) = \Phi^{-1}(u) = \beta_2(u)$, where $\Phi$ is the standard normal distribution function. 
For the CRK test, I estimated \eqref{eq:mcqr} separately for each cluster, obtained $q$ estimates of $\beta_1$ and applied Algorithm \ref{al:within} with 1,000 new draws from $\G$ for each Monte Carlo replication.

To the best of my knowledge, there are no other methods of inference designed specifically for quantile functions or Kolmogorov-Smirnov statistics in data with few large clusters. I therefore compare the CRK test to inference with the wild gradient bootstrap \citep{hagemann2017}, a cluster-robust version of the bootstrap that requires the number of clusters $q\to\infty$ for consistency. The wild gradient bootstrap is the default option for cluster-robust inference in the \texttt{quantreg} package in \texttt{R}. I use the package default settings with Mammen bootstrap weights and 200 bootstrap simulations. 
Alternative analytical methods for cluster-robust inference in quantile regressions exist but can only perform pointwise inference because the QR process as $q\to\infty$ generally has an analytically intractable distribution. \citet{hagemann2017} shows that the wild gradient bootstrap can conduct uniform inference on quantile regression functions and that it outperforms other methods for pointwise inference in this context. 
However, \citet{hagemann2017} notes that size distortions can occur when fewer than 20 clusters are present. I therefore focus on this situation in the following.

\begin{figure}
\centering
\resizebox{\textwidth}{!}{
\input{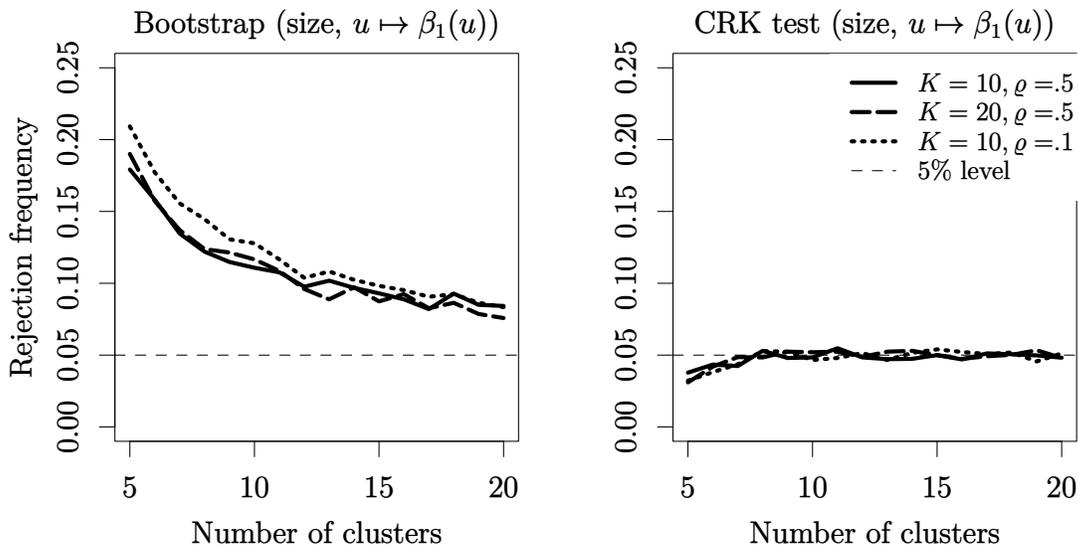}
}	
\caption{Rejection frequencies in Example \ref{ex:qr2} of a true null $H_0\colon \beta_1(u) = 0$ for all $u$ as a function of the number of clusters for the bootstrap (left) and the CRK test (right) with (i)~$K=10$ neighborhoods per cluster with intra-neighborhood correlation $\varrho = .5$ (solid lines), (ii)~$K=20$ with $\varrho = .5$ (long-dashed), and (iii)~$K=10$ with $\varrho = .1$ (dotted). Short-dashed line equals nominal level $.05$.}\label{f:ex41-fig-size}
\end{figure}

Figure~\ref{f:ex41-fig-size} shows the rejection frequencies of a true null hypothesis $H_0\colon \beta_1(u) = 0$ for all $u$ as a function of the number of clusters $q\in \{5,6,\dots, 20\}$ for the wild gradient bootstrap (left) and the CRK test (right) at the 5\% level (short-dashed line). The figure shows rejection frequencies in 5,000 Monte Carlo replications for each horizontal coordinate with (i)~$K=10$ neighborhoods per cluster with intra-neighborhood correlation $\varrho = .5$ (solid lines), (ii)~$K=20$ with $\varrho = .5$ (long-dashed), and (iii)~$K=10$ with $\varrho = .1$ (dotted). Both methods were faced with the same data and I estimated $\beta_1$ at $u = .1, .2, \dots, .9$ for both methods. As can be seen, the wild gradient bootstrap over-rejected mildly with 20 clusters but over-rejected substantially for smaller numbers of clusters. It exceeded a 10\% rejection rate if only 12 clusters were available. With 5 clusters, the wild gradient bootstrap falsely discovered an effect in up to 20.9\% of all cases ($K=10, \varrho = .1$). In contrast, the CRK test rejected at or slightly below nominal level for all $q$ and all configurations of $K$ and $\varrho$.

I also experimented with a large number of alternative DGPs under the null. I considered (not shown) larger neighborhoods, different values of $\varrho$, different spatial dependence structures such as (spatial) autoregressive models, and different distributions for $X_{i,j,k}$. However, I found that these changes had little qualitative impact on the results described in the preceding paragraph or in \citet{hagemann2017}. The wild gradient bootstrap generally performed very well but experienced size distortions with fewer than 20 clusters. The CRK test rejected at or slightly below nominal level in all situations I investigated.

I now turn to the behavior of the test under the alternative. I repeated the experiment but now tested the incorrect null hypothesis $H_0\colon \beta_2(u) = 0$ for all $u\in\mathcal{U}$. Figure~\ref{f:ex41-fig-pow} shows the rejection frequencies of this null against the alternative $H_1\colon \beta_2(u) > 0$ for some $u\in\mathcal{U}$, where $\mathcal{U}$ was either $(0, 1)$ (black) or $(.5, 1)$ (grey). The null hypothesis is false in both situations but the case where $\mathcal{U} = (0, 1)$ is more challenging because $\beta_2(u) < 0$ for all $u < .5$ so that estimates below the median provide evidence in the direction away from the alternative. I again considered (i)~$K=10$ neighborhoods per cluster with intra-neighborhood correlation $\varrho = .5$ (solid lines), (ii)~$K=20$ with $\varrho = .5$ (long-dashed), and (iii)~$K=10$ with $\varrho = .1$ (dotted). As could be expected, the bootstrap rejected a large fraction of null hypotheses mostly because it was unable to control the size of the test. However, it had high power when the number of clusters was above 20 and the size distortions disappeared (not shown). The CRK test had high power while maintaining size control even when the number of clusters was below 20. For example, at $q=12$ it detected a deviation from the null between 22.5\% ($K=10, \varrho=.5, \mathcal{U} = (0,1)$) and 84.26\% ($K=20, \varrho=.5, \mathcal{U} = (.5,1)$) of all cases. More generally, additional clusters, lower intra-cluster dependence, and additional neighborhoods per cluster increased the power of the CRK test. \sqed

\begin{figure}
\centering
\resizebox{\textwidth}{!}{
\input{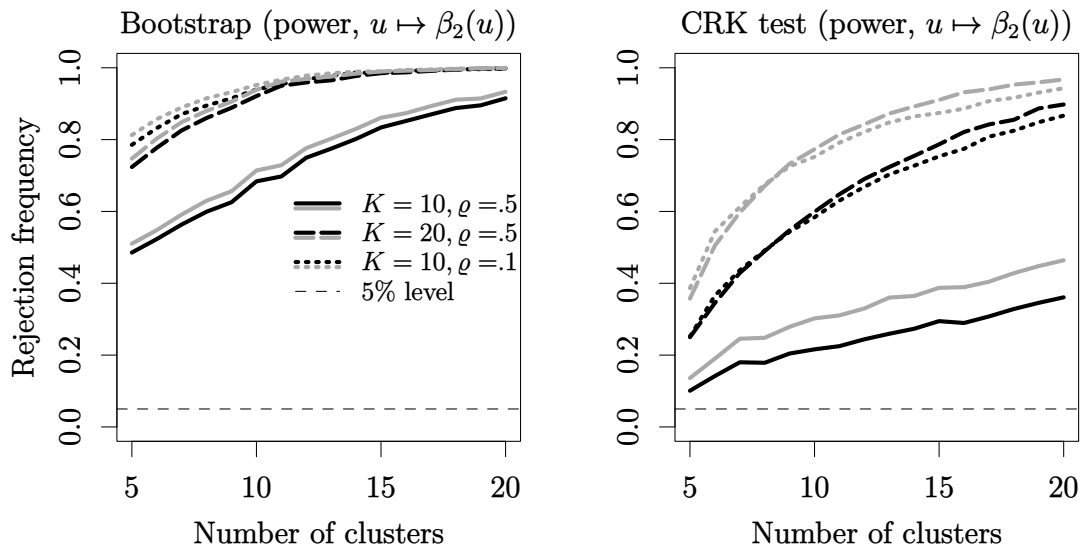}
}	
\caption{Rejection frequencies in Example \ref{ex:qr2} of false nulls $H_0\colon \beta_2(u) = 0$ for $u > .5$ (grey) and $H_0\colon \beta_2(u) = 0$ for all $u$ (black) as a function of the number of clusters for the bootstrap (left) and the CRK test (right) with (i)~$K=10$ neighborhoods per cluster with intra-neighborhood correlation $\varrho = .5$ (solid lines), (ii)~$K=20$ with $\varrho = .5$ (long-dashed), and (iii)~$K=10$ with $\varrho = .1$ (dotted). 
}
\label{f:ex41-fig-pow}
\end{figure}

\end{example}

\begin{example}[Quantile treatment effects, cont.]\label{ex:qte2}
	For this experiment, I reuse the setup of Example \ref{ex:qr2} but replace the variable $X_{i,j,k}$ with a cluster-level treatment indicator $D_j$ that equals one if cluster $j$ received treatment and equals zero otherwise.  I randomly assign $q_1 = \lfloor q/2 \rfloor$ clusters to treatment and $q_0 = \lceil q/2 \rceil$ to control. The coefficient of interest is $\delta$ in 
	\begin{align*}
Q(u\mid D_{j}) = \beta_0(u) + \delta(u)D_j + \beta_2(u)Z_{i,j}.
\end{align*}
I do not assume that pairings are predetermined and therefore use the adjusted $p$-values of the CRK test from Algorithm \ref{al:between}. For each Monte Carlo replication, I drew a collection $\mathcal{I}$ with $|\mathcal{I}| = 50$ from $\mathcal{H}$ without replacement. The CRK test with unknown cluster parings requires $\alpha = .05 > 1/2^{q_1\wedge q_0-1}$ to have power, which is satisfied here as long as $q\geq 12$. I therefore restrict $q$ to be between 12 and 20. All other parameters of the experiment are exactly as in Example \ref{ex:qr2}.

\begin{figure}
\centering
\resizebox{\textwidth}{!}{
\input{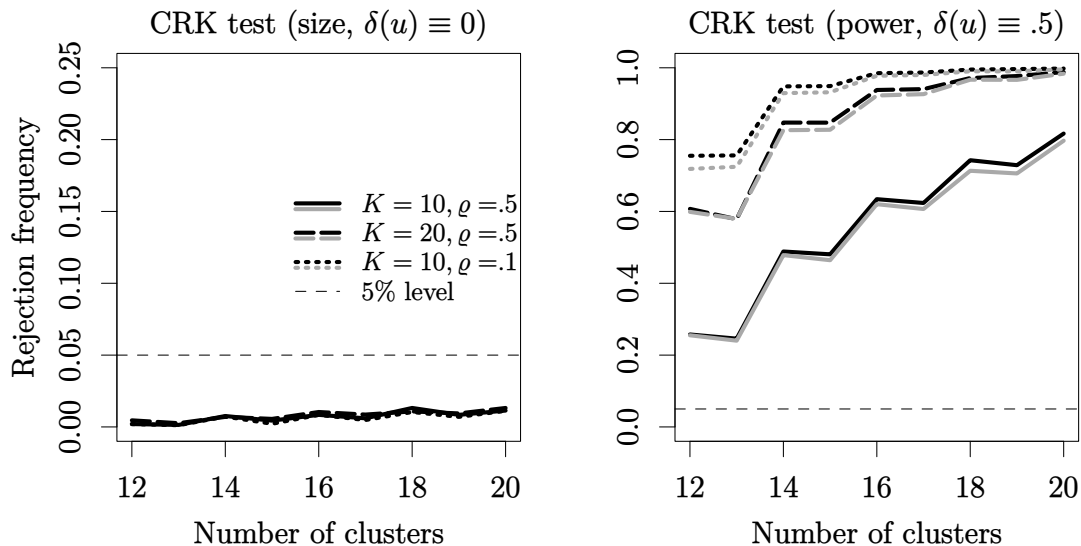}
}	
\caption{Rejection frequencies  in Example \ref{ex:qte2} of a true null (left) $H_0\colon \delta(u) = 0$ for all $u$ and false nulls (right) $H_0\colon \delta(u) = 0$ for $u > .5$ (grey) and $H_0\colon \delta(u) = 0$ for all $u$ (black) as a function of the number of clusters for the CRK test when cluster pairings are not known with (i)~$K=10$ neighborhoods per cluster with intra-neighborhood correlation $\varrho = .5$ (solid lines), (ii)~$K=20$ with $\varrho = .5$ (long-dashed), and (iii)~$K=10$ with $\varrho = .1$ (dotted). 
}
\label{f:ex42-fig}
\end{figure}

The left panel of Figure~\ref{f:ex42-fig} shows the rejection frequencies of a true null hypothesis $H_0\colon \delta(u) = 0$ for all $u$ in 5,000 Monte Carlo experiments per horizontal coordinate as $q$ increases. I again considered (i)~$K=10$ neighborhoods per cluster with intra-neighborhood correlation $\varrho = .5$ (solid lines), (ii)~$K=20$ with $\varrho = .5$ (long-dashed), and (iii)~$K=10$ with $\varrho = .1$ (dotted). As can be seen, adjusting the CRK test for unknown cluster pairings results in a markedly more conservative test relative to an unadjusted test from Figure~\ref{f:ex41-fig-size}. However, as the right panel of Figure~\ref{f:ex42-fig} shows, this did not translate into poor power under the alternative. When I repeated the experiment with $\delta(u)\equiv .5$, the CRK test with identification across clusters had no problem detecting that neither $H_0\colon \delta(u)$ for all $u\in (0,1)$ (black) nor $H_0\colon \delta(u)$ for $u > .5$ (grey) were true. Compared to Example \ref{ex:qr2}, the alternative where $\mathcal{U} = (0,1)$ rejects slightly more nulls because now every $u$ provides evidence against the null.

A noteworthy feature of the right panel of Figure~\ref{f:ex42-fig} is the ``zig-zag'' pattern in the rejection frequencies. The reason for this pattern is the treatment assignment mechanism. If $q=12$, then $q_1 = 6$ clusters receive treatment and $q_0 = 6$ do not. If $q=13$, then again $6 = \lfloor 13/2\rfloor$ clusters receive treatment but now $7 = \lceil 13/2 \rceil$ do not. Algorithm~\ref{al:between} uses a large number of potential pairings of treatment to control for inference but effectively reduces the number of clusters to $\min\{q_1,q_0\}$. In this experiment, inference with $6+7$ clusters is therefore effectively the same as inference with $6+6$ clusters, which explains the similar performance of the test at $q$ and $q-1$ when $q$ is odd.

I also experimented with alternative methods for combining the $p$-values in Algorithm~\ref{al:between}. For this, I repeated the experiment in the right panel of Figure~\ref{f:ex42-fig} (not shown) but replaced the left-hand side of \eqref{eq:avgpval2} with either the standard Bonferroni correction or the geometric average of the $p$-values times $\exp(1)$. The latter method is due to \citet{mattner2012} and discussed in detail in \citet{vovkwang2020}. 
At $q=12$ with $K=20$ and $\varrho = .5$, Bonferroni rejected the false null $H_0\colon \delta(u) = 0$ for all $u$ in none of the 5,000 simulations, \citeauthor{mattner2012}'s method rejected in 39.44\% of all simulations, and Algorithm \ref{al:between} rejected in 60.76\% of all simulations in the same data. At $q=20$, \citeauthor{mattner2012}'s method and Algorithm \ref{al:between} rejected in about 99\% of all cases. Bonferroni improved substantially and rejected in 96\% of all cases. I conducted a large number of additional experiments but the results remained the same. The Bonferroni method had by far the lowest power. Relative to Algorithm \ref{al:between}, \citeauthor{mattner2012}'s method had significantly lower power when the number of clusters was small but both alternatives caught up as this number increased. Neither Bonferroni nor \citeauthor{mattner2012}'s method improved the power of Algorithm \ref{al:between} in any of my experiments.

Finally, I compare Algorithm~\ref{al:between} to the ideal situation where an equal number of treated and control clusters are paired through a pilot experiment or pre-analysis plan. For this, I repeated the experiment in Figure~\ref{f:ex42-fig} with a single, randomly chosen pairing. At $q_1 = q_0 =8$, $K=10$, and $\varrho = .5$, a test with pre-specified pairs rejected a false $H_0\colon \delta(u) = 0$ for all $u$ in 82.16\% of all cases as opposed to the 63.40\% achieved by Algorithm \ref{al:between}. In the same experiment with $q_1 = q_0 =10$, the test with pre-specified pairs rejected 91.90\%  of all false nulls and Algorithm \ref{al:between} rejected 81.70\% of all false nulls. However, there are 8! = 40,320 potential ways of paring $q_1 = 8$ treated clusters and $q_0=8$ control clusters. If $q_1=q_0=10$, there are 3,628,800 ways. Each separate set of pairs could be potentially selected for the test. If a researcher did not pre-specify pairs and instead searched over potential pairs to discover a significant result, the test would quickly lose size control. At $q_1 = q_0 =6$, $K=10$, and $\varrho = .5$, a correct null hypothesis was erroneously rejected in 8.14\% of all cases when the lowest $p$-value among three potential pairings was used. If the researcher chooses the lowest $p$-value among ten potential matches, then a non-existent effect showed up as statistically significant in 13.96\% of all cases in a test with 5\% nominal level. Algorithm \ref{al:between} does not choose among these tests and completely avoids this loss of size control. \sqed
\end{example}

\begin{example}[Placebo interventions in Project STAR]\label{ex:star}
	In this example, I revisit a challenging placebo exercise of \citet[][Experiment~5.1]{hagemann2017} in data from the first year of the Tennessee \emph{Student/Teacher Achievement Ratio} experiment, known as Project STAR. Details about the data can be found in \citetalias{wordetal1990} and \citet{graham2008}. I only provide a brief summary. 

In 1985, incoming kindergarten students in 79 project schools were randomly assigned to small classes (13-17 students) or regular-size classes (22-25 students) with or without a teacher's aide. 
Each of the project schools was required to have at least one of each kindergarten class type. The outcome is standardized student performance on the \emph{Stanford Achievement Test} (SAT) in mathematics and reading administered at the end of the school year. The raw test scores are standardized as in \citet{krueger1999}. He finds across several mean regression models that students in small classes perform about five percentage points better on average than students in regular classrooms. (Assigning teachers aides had no effect uniformly across specifications and I do not consider such classes in the following.) \citet{jacksonpage2013} and \citet{hagemann2017} document similar effects in quantile regressions but show that the effects are smaller for students near the bottom and the very top of the conditional outcome distribution and larger near the center of the distribution. For example, in the model 
\begin{align}\label{eq:star1}
Q_{Y_{i,j}}(u\mid X_{i,j}) = \beta_0(u) + \delta(u) \mathit{small}_{i,j} + \beta_2(u)^T Z_{i,j}
\end{align}
where the treatment dummy $\mathit{small}$ indicates whether the student was assigned to a small class and $Z$ contains school dummies, the effect of being in a small class relative to a regular class varies between 2.78 percentage points at the 10th percentile to 7.23 percentage points at the 60th percentile. \citet{jacksonpage2013} hypothesize that this heterogeneity could be attributed to varying levels of student motivation to take advantage of increased individual attention from a teacher.

For the placebo experiment, I removed all small classes from the sample and only kept the 16 schools that had two regular-size classes without aide. In each of these 16 schools, I then randomly assigned one of the regular-size classes the treatment indicator $\mathit{small}=1$. This mimics the random assignment of class sizes within schools in the original sample, even though in this case no student actually attended a small class. I clustered at the classroom level and applied the CRK test as in Algorithm~\ref{al:within} by running 16 separate quantile regressions, one for each school, on a constant and $\mathit{small}$ to get 16 separate estimates of $\delta$. The fixed effects as in \eqref{eq:star1} are not needed here because the constant can vary freely by school in these quantile regressions. Algorithm~\ref{al:within} applies here because each school in this experiment has only one class with $\mathit{small} = 1$ and one class with $\mathit{small} = 0$. (If multiple small classes per school were available, then Algorithm \ref{al:between} could be used instead.) For the wild gradient bootstrap, I reran the QR in \eqref{eq:star1} in the placebo data and again clustered at the classroom level. For both methods, I tested at the 5\% level the correct null hypothesis that $H_0\colon \delta(u) = 0$ jointly at $u \in \{.1, .2, \dots, .9 \}$ against the alternative that $\delta$ is positive. 

The rejection frequencies in `size' column in Table \ref{t:placebo} show the outcome of repeating the placebo assignment 1,000 times. As can be seen, the CRK test provided a nearly exact test but the bootstrap over-rejected somewhat. The over-rejection for the bootstrap here was documented by \citet{hagemann2017} and  
can be attributed to the very small number of clusters available in the placebo sample vis-\`a-vis the large number of clusters needed for the consistency of the bootstrap.

\begin{table}\caption{Rejection frequencies of $H_0\colon \delta(u) = 0$ for all $u$ in placebo interventions in Project STAR for the CRK test and the wild gradient bootstrap at 5\% level}\label{t:placebo}
{\centering
\scalebox{1}{
\begin{tabular}{lp{0cm}cp{0cm}cccccc}	
\hline
& & size & &\multicolumn{6}{c}{power} \\ \cline{3-3}\cline{5-10}
& &$\delta = 0$ & &$\delta = 2$ & $\delta = 3$  &$\delta = 4$  & $\delta = 5$   & $\delta = 6$ & $\delta = 7$\\ \hline
CRK test & &.043 & &.122 & .161  &.212  & .318   & .379 & .478 \\
Bootstrap & &.091 & &.233 & .316 &.428   & .580   & .691 & .814 \\ \hline
\end{tabular}	
}}
\end{table}

I also investigated power by increasing the percentile scores of all students in the randomly drawn small classes of the placebo experiment by $\delta \in \{2,3,4,5,6 ,7\}$ percentage points. These increases are of the same or smaller magnitude as the estimated quantile treatment effects in the actual sample. Then I tested the incorrect hypothesis $H_0\colon \beta_1(u) = 0$ for all $u$ with the same experimental setup as before. The results are shown in `power' column of Table \ref{t:placebo}. As can be seen, the CRK test was able to reliably detect effects for moderate deviations from the null hypothesis. The wild gradient bootstrap rejected more often, but this was likely caused by its tendency to over-reject in this data set.\sqed
\end{example}

\section{Conclusion}
I introduce a generic method for inference on quantile and regression quantile processes in the presence of a finite number of large and arbitrarily heterogeneous clusters. The method asymptotically controls size by generating statistics that exhibit enough distributional symmetry such that randomization tests can be applied. This randomization test can even be asymptotically similar in empirically relevant situations. The test does not require ex-ante matching of clusters, is free of user-chosen parameters, and performs well at conventional significance levels with as few as five clusters. The main focus on the paper is inference on quantile treatment effects and quantile difference in differences but the method applies more broadly.  Numerical examples and an empirical application are provided.

\bibliographystyle{chicago}
\bibliography{qspec.bib}

\appendix

\section{Proofs}

\begin{proof}[Proof of Theorem \ref{th:hoeffding}]
Denote the inverse element of $g\in\G$ by $g^{-1}$ and the identity element by $\id$. The proof uses the inverse $g^{-1}$ of $g\in\G$ to clarify its role in the argument. However, note that inverting $g$ is redundant for this particular $\G$ because it satisfies $g^{-1} = g$. I first argue that $(T(gX))_{g\in\G}$ satisfies $(T(gX))_{g\in\G} \sim (T(g\tilde{g}^{-1}X))_{g\in\G}$, where $\tilde{g}$ is an arbitrary element of $\G$. For this, I show that both quantities must have the same distribution at continuity points and that $(T(gX))_{g\in\G}$ has a continuous distribution. I then argue that $(T(gX))_{g\in\G} \sim (T(g\tilde{g}^{-1}X))_{g\in\G}$ is enough for \citeauthor{hoeffding1952}'s (\citeyear{hoeffding1952}) argument to go through.

Take a finite grid of points $\mathcal{U}_m := \{i/m : i=0,1,\dots,m \} \cap \mathcal{U}$. Then every $u\in\mathcal{U}$ is a limit of a sequence in $\mathcal{U}_m$. Let $x\mapsto T_m(x) = \sup_{u\in\mathcal{U}_m }\sum_{j=1}^q x_j(u)/q$. Uniform continuity implies $T_m(x) \to T(x)$ and $(T_m(gX))_{g\in\G} \to (T(gX))_{g\in\G}$ almost surely and therefore also $(T_m(gX))_{g\in\G} \leadsto (T(gX))_{g\in\G}$.  Independence and $\prob(X_j(u) = 0) = 0$ ensure that $\sum_{j=1}^q X_j(u)/q$ has a continuous distribution at every $u$. Because $X$ is separable, $T(gX) = \sup_{u\in\mathcal{U}\cap \mathbb{Q}} \sum_{j=1}^q X_j(u)/q$, where $\mathbb{Q}$ are the rationals. Conclude that $(T(gX))_{g\in\G}$ has a continuous distribution because for arbitrary $t_g\in\mathbb{R}$, \[ \prob\bigcap_{g\in\G} \{T(gX) = t_g\} \leq \prob\Biggl(\sup_{u\in\mathcal{U}\cap \mathbb{Q}} \frac{1}{q}\sum_{j=1}^q X_j(u) = t_{\id}\Biggr)\leq \bigcup_{u\in\mathcal{U}\cap \mathbb{Q}}\prob\Biggl(\frac{1}{q}\sum_{j=1}^q X_j(u) = t_{\id}\Biggr) \] and the extreme right-hand side equals zero.  Finite-dimensional distributional invariance implies that $(T_m(gX))_{g\in\G}$ and $(T_m(g\tilde{g}^{-1}X))_{g\in\G}$ have the same distribution for every $\tilde{g}\in\G$. Because $(T_m(gX))_{g\in\G}\leadsto (T(gX))_{g\in\G}$, it must also be true that $(T_m( g\tilde{g}^{-1}X))_{g\in\G}\leadsto (T(gX))_{g\in\G}$ and $(T_m(g\tilde{g}^{-1}X))_{g\in\G}\leadsto (T(g\tilde{g}^{-1}X))_{g\in\G}$. Conclude from continuity that $(T(gX))_{g\in\G}$ and $(T(g\tilde{g}^{-1}X))_{g\in\G}$ have the same distribution for every $\tilde{g}\in\G$. These two random vectors are of the form
\[ (T(X),\dots, T(gX), \dots, T(\tilde{g}X), \dots) \sim (T(\tilde{g}^{-1}X),\dots, T(g\tilde{g}^{-1}X), \dots, T(\id X), \dots). \]
Because $T^{1-\alpha}(X,\G \tilde{g}^{-1} ) = T^{1-\alpha}(X,\G) = T^{1-\alpha}(\tilde{g}X,\G)$, this implies $\varphi_{\alpha}(X, \G ) \sim \varphi_{\alpha}(\tilde{g}X, \G )$, where $\varphi_{\alpha}(\tilde{g}X, \G )$ is the test function $\varphi_{\alpha}(X, \G )$ computed with $\tilde{g}X$ instead of $X$. Because $\tilde{g}\in\G$ was arbitrary, conclude $\ev \sum_{g\in\G} \varphi_{\alpha}(gX, \G ) = \ev \varphi_{\alpha}(X, \G )|\G| $. The same argument as the finite-dimensional case now yields $\ev \varphi_{\alpha}(X, \G )\leq \alpha$.

For the randomized test decision introduced in \eqref{eq:randtest}, the arguments so far also imply that $\phi_{\alpha}(X, \G ) \sim \phi_{\alpha}(gX, \G )$ for every $g\in\G$. Use the definition on $a(X)$ to see that $\sum_{g\in\G} \phi_{\alpha}(gX, \G ) =  |\{g\in\G : T(gX) > T^{1-\alpha}(X,\G )| + a(X) |\{g\in\G : T(gX) = T^{1-\alpha}(X,\G )| = \alpha |\G|$
and therefore \[ \prob (\phi_{\alpha}(X, \G) \geq V) = \ev \phi_{\alpha}(X, \G )= \frac{1}{|\G|}\sum_{g\in\G}  \ev \phi_{\alpha}(gX, \G ) = \frac{1}{|\G|} \ev\sum_{g\in\G}  \phi_{\alpha}(gX, \G ) = \alpha, \] as desired.
\end{proof}

\begin{proof}[Proof of Theorem \ref{th:asysize}]
For $x, x'\in \ell^\infty(\mathcal{U})^q$ and every $g\in\G $, sub-additivity and monotonicity give
\[ T(g x) - T(g x') \leq \sup_{u \in \mathcal{U}}\frac{1}{q}\Biggl(\sum_{j=1}^q g_j \bigl(x_j(u) - x_j'(u)\bigr)\Biggr) \leq  \sup_{u \in \mathcal{U}} \frac{1}{q}\sum_{j=1}^q \bigl|x_j(u) - x_j'(u)\bigr|.  \] The far right of the display is at most $|x - x'|_\mathcal{U}/\sqrt{q}$. Reverse the roles of $x$ and $x'$ to conclude $ |T(g x) - T(g x')|^2 \leq |x - x'|^2_\mathcal{U}/q$ for every $g\in\G $ and therefore \[ \bigl|\bigl(T(g x) - T(g x')\bigr)_{g\in\G }\bigr| \leq \sqrt{2^q/q}|x - x'|_\mathcal{U}. \] Let $|x - x'|_\mathcal{U}\to 0$ to deduce that $x \mapsto (T(g x))_{g\in\G }$ a continuous map from $\ell^\infty(\mathcal{U})^q$ to $\mathbb{R}^{|\G |}$ with respect to the sup-norm. Because $X_n \leadsto X$, the continuous mapping theorem implies $(T(g X_n))_{g\in\G } \leadsto (T(g X))_{g\in\G }$. 

Order $\G $ so that the identity action $g = (1, \dots, 1)$ is the first element. Define \begin{equation}\label{eq:boundaryset}
 B_\alpha = \Bigl\{ (t_1, t_2, \dots, t_{|\G |}) : |\{2\leq i\leq |\G | : t_i < t_1 \}| \geq \lceil (1-\alpha)|\G |\rceil \Bigr\}
 \end{equation} 
as the set of all vectors where the first element of the vector exceeds at least $\lceil (1-\alpha)|\G |\rceil$ of the remaining elements. Because  only the relative ranking of $(T(g X))_{g\in\G }$ enters the test decision, the test rejects if and only if $(T(g X))_{g\in\G } \in B_\alpha$. Conclude that $\prob(T(X) > T^{\alpha}(X, \G )) = \prob((T(g X))_{g\in\G } \in B_\alpha)$. The boundary $\partial B_\alpha$ of $B_\alpha$ can be expressed as \[ \partial B_\alpha = \bigcup_{j \geq 1} \Bigl\{ (t_1, t_2, \dots, t_{|\G |}) : |t_1 = t_i| = j, |\{2\leq i\leq |\G | : t_i < t_1 \}| = \lceil (1-\alpha)|\G |\rceil - j \Bigr\} \] and therefore $\partial B_\alpha \subset \bigcup_{j\geq 1} \{ (t_1, t_2, \dots, t_{|\G |}) : |t_1 = t_i| = j\}$. By the portmanteau lemma, $\prob((T(g X_n))_{g\in\G } \in B_\alpha) \to \prob((T(g X))_{g\in\G } \in B_\alpha)$ as long $\partial B_\alpha$ satisfies $\prob((T(g X))_{g\in\G } \in \partial B_\alpha) = 0.$ The goal is therefore to show that \[ \prob\biggl((T(g X))_{g\in\G } \in \bigcup_{j\geq 1} \{ (t_1, t_2, \dots, t_{|\G |}) : |t_1 = t_i| = j\}\biggr) = 0, \]
i.e., $(T(g X))_{g\in\G }$ has no ties with probability one.

The main difficulty here is that each component of $(T(g X))_{g\in\G }$ is dependent, so the preceding display does not follow from smoothness of the marginals of $(T(g X))_{g\in\G }$. Instead, for $u,u' \in \mathcal{U}$ and $g \neq g'$, write
\[ \sum_{j=1}^q g_j X_j(u) - \sum_{j=1}^q g_j' X_j(u')  = (g, -g')^T (X(u), X(u')) \]
Because $X$ is a Gaussian process, it follows that $(X(u), X(u'))$ is a jointly Gaussian vector and therefore $(g, -g')^T (X(u), X(u'))$ is a normally distributed random variable. 

If $u = u'$ or $u \neq u'$ but $X(u) = X(u')$, then $g\neq g'$ guarantees that $\sum_{j=1}^q g_j X_j(u) - \sum_{j=1}^q g_j' X_j(u) = \sum_{j=1}^q (g_j-g_j') X_j(u)$ has non-zero variance. Hence, suppose $u \neq u'$ and $X(u) \neq X(u')$. Let $c(u, u') = \ev X(u) X(u')$ be the covariance function and note that $(g, -g')^T (X(u), X(u'))$ is zero with positive probability if and only if $(g, -g')^T c(u, u') (g, -g') = 0$. Because the elements of $X$ are independent, the covariance function satisfies 
\[ (g, -g')^T c(u, u') (g, -g') = \sum_{j=1}^n c_{jj}(u, u) + \sum_{j=1}^n c_{jj}(u', u') - 2\sum_{j=1}^n g_j g_j' c_{jj}(u, u'). 
\] Apply the Cauchy-Schwarz inequality to the right-hand side to deduce \[ 0 = (g, -g')^T c(u, u') (g, -g') \geq \sum_{j=1}^n \bigl (c_{jj}(u, u) - c_{jj}(u', u') \bigr)^2, \] which implies $\var X_j(u) = \var X_j(u')$ for $1\leq j \leq q$. It follows that 
\[ 0 =  \sum_{j=1}^n \bigl (c_{jj}(u, u) - g_j g_j' c_{jj}(u, u')\bigr) 
\] Apply the Cauchy-Schwarz inequality again to see that every covariance must be non-zero because $c_{jj}(u, u) > 0$ and either $c_{jj}(u, u') = c_{jj}(u, u)$ or $c_{jj}(u, u') = -c_{jj}(u, u)$. This implies that either $X_{j}(u) = X_{j}(u')$ or $X_{j}(u) = -X_{j}(u')$. Because $g\neq g'$, $X(u) = X(u')$ is impossible and at least one $j$ must satisfy $X_{j}(u) = -X_{j}(u')$, which is ruled out by assumption. Conclude \[ \sum_{j=1}^q g_j X_j(u) \neq \sum_{j=1}^q g_j' X_j(u')\] almost surely for all $u, u'\in \mathcal{U}$ and all $g\neq g'$. Because $\mathcal{U}$ is compact and $X$ has continuous sample paths, this ensures \[ T(gX) - T(g'X) = \max_{u\in \mathcal{U}} \sum_{j=1}^q g_j X_j(u) - \max_{u\in \mathcal{U}} \sum_{j=1}^q g_j' X_j(u) \neq 0 \] for almost every sample path unless $g=g'$.
\end{proof}

\begin{proof}[Proof of Theorem \ref{th:size}] If $H_0$ is true, then scale invariance implies $\varphi_{\alpha}(\hat{\delta} - \delta_0 1_q, \G ) = \varphi_{\alpha}(X_n, \G )$. Assumption~\ref{as:weakc} and Theorem~\ref{th:asysize} yield $\ev \varphi_{\alpha}(\hat{\delta} - \delta_0 1_q, \G ) \to \ev \varphi_{\alpha}(X, \G )$. $\ev \varphi_{\alpha}(X, \G )\leq \alpha$ holds because $X$ satisfies the conditions of Theorem~\ref{th:hoeffding}.
\end{proof}

\begin{proof}[Proof of Theorem \ref{th:power}]
Suppose $\delta = \delta_0 + \lambda/\sqrt{n}$ so that $X_n + \lambda 1_q \wto{\delta} X + \lambda 1_q$. As in the proof of Theorem \ref{th:size}, joint continuity of the map $x \mapsto (T(gx))_{g\in\G }$ implies $(T(g(X_n + \lambda 1_q)))_{g\in\G }\wto{\delta} (T(g(X + \lambda 1_q)))_{g\in\G }$. With $B_\alpha$ as defined in \eqref{eq:boundaryset}, I only have to show that $\prob((T(g(X + \lambda)))_{g\in\G } \in \partial B_\alpha) = 0$ to conclude $\prob(T(X_n + \lambda) > T^{\alpha}(X_n + \lambda, \G )) \to \prob(T(X + \lambda) > T^{\alpha}(X + \lambda, \G ))$.

The boundary has probability zero if $(T(g(X + \lambda)))_{g\in\G }$ has no ties. For $u,u' \in \mathcal{U}$ and $g \neq g'$, write
\[ \Biggl[ \sum_{j=1}^q g_j X_j(u) - \sum_{j=1}^q g_j' X_j(u') \Biggr] + \lambda(u)\sum_{j=1}^q g_j - \lambda(u')\sum_{j=1}^q g_j', \]
to see from the proof of Theorem \ref{th:size} that the term in square brackets is nonzero almost surely for all $u, u'\in \mathcal{U}$ and all $g\neq g'$. Because the expression in square brackets is normally distributed with mean zero, it cannot take on any fixed nonzero value with positive probability. The remainder of the preceding display is constant. Conclude that the preceding display is nonzero almost surely for all $u, u'\in \mathcal{U}$ and all $g\neq g'$. As in the proof of Theorem \ref{th:size}, this implies $T(g(X + \lambda)) \neq T(g'(X + \lambda))$ almost surely unless $g\neq g'$.

I will now develop a lower bound on $\prob(T(X + \lambda 1_q) > T^{\alpha}(X + \lambda  1_q, \G ))$. Because the original statistic cannot exceed the largest order statistic, monotonicity implies 
\begin{align*}
 \prob\bigl(T(X + \lambda  1_q) > T^{1-\alpha}(X + \lambda  1_q, \G )\bigr) &\geq \prob\bigl(T(X + \lambda  1_q) > T^{(|\G |-1)}(X + \lambda  1_q, \G )\bigr) \\
 &=\prob\Bigl(T(X + \lambda  1_q) = \max_{g\in\G }T\bigl(g(X + \lambda  1_q)\bigr)\Bigr)
\end{align*}
and the right-hand side is at most \[ \prob\biggl(\Bigl\{T(X + \lambda  1_q) = \max_{g\in\G }T\bigl(g(X + \lambda  1_q)\bigr) \Bigr\}, \bigcap_{j=1}^q \Bigl\{ \inf_{u\in\mathcal{U}}(X_j(u) + \lambda(u)) \geq 0  \Bigr\} \biggr). \]
If $\inf_{u\in\mathcal{U}}(X_j(u) + \lambda(u)) \geq 0$ for $1\leq j\leq q$, then $T(X + \lambda  1_q) = \max_{g\in\G }T(g(X + \lambda  1_q))$ because $T$ cannot be increased by making large negative values positive through multiplication by $-1$. By independence and symmetry of the Gaussian processes, conclude that the preceding display equals \[ \prod_{j=1}^q \prob\Bigl (\inf_{u\in\mathcal{U}}\bigl(X_j(u) + \lambda(u)\bigr) \geq 0\Bigr) = \prod_{j=1}^q \prob\Bigl (\sup_{u\in\mathcal{U}}\bigl(X_j(u) - \lambda(u)\bigr) \leq 0\Bigr). \] 
 Because $\sup (f - f') \geq \sup f - \sup f'$ for arbitrary $f,f'$, this cannot exceed \[ \prod_{j=1}^q \prob\Bigl (\sup_{u\in\mathcal{U}}X_j(u) \leq \sup_{u\in\mathcal{U}} \lambda(u)\Bigr) \geq \prod_{j=1}^q \Bigl( 1 - e^{-[\sup_{u} \lambda(u) - \ev \sup_{u} X_j(u)]^2/2 \sup_{u} \ev X^2_j(u)}   \Bigr) \] by the Borell-TIS inequality as long as $\sup_{u} \lambda(u) > \ev \sup_{u} X_j(u)$. In that case, the right-hand side of the preceding display is strictly positive, as required.

Suppose $\delta = \delta_0 + \lambda 1_q$. We have $\hat{\delta} - \delta_0 = X_n/\sqrt{n} + \lambda 1_q$ with $X_n/\sqrt{n}\leadsto 0$, and by arguments as in the proof of Theorem~\ref{as:weakc}, the continuous mapping theorem yields $(T(g(\hat{\delta} - \delta_0 1_q)))_{g\in\G } \leadsto (T(g\lambda))_{g\in\G }$. Monotonicity implies \[ \ev \varphi_{1-\alpha}(\hat{\delta}-\delta_0 1_q,\G ) \geq \prob\bigl( T(\hat{\delta}-\delta_0 1_q) > T^{(|\G |-1)}(\hat{\delta}-\delta_0 1_q,\G ) \bigr) \] As before, use a set of the form \[ B = \Bigl\{ (t_1, t_2, \dots, t_{|\G |}) : |\{2\leq i\leq |\G | : t_i < t_1 \}| \geq |\G | - 1 \Bigr\} \] to write $\prob(T(\lambda) > T^{(|\G |-1)}(\lambda, \G )) = \prob((T(g \lambda))_{g\in\G } \in B) = 1\{(T(g \lambda))_{g\in\G } \in B\}$. The boundary $\partial B$ is contained in the set \[ \bigcup_{j \geq 1} \Bigl\{ (t_1, t_2, \dots, t_{|\G |}) : |t_1 = t_i| = j\Bigr\}. \] Because $T(g\lambda) = \sup_{u\in\mathcal{U}}  \lambda(u)\sum_{j=q}^q g_j/q$ with $\lambda \geq 0$ and $\sup_{u\in\mathcal{U}} \lambda(u) > 0$, we have $T(\lambda) > T(g\lambda)$ for all $g\neq \mathrm{id} := (1, \dots, 1)$. Hence, there are no ties with the first element of $(T(g \lambda))_{g\in\G }$ and $1\{(T(g \lambda))_{g\in\G } \in \partial B\} = 0$. Conclude from the portmanteau lemma that $T(\hat{\delta} - \delta_0 1_q) - T^{(|\G |-1)}(\hat{\delta} - \delta_0 1_q, \G ) \leadsto \sup_{u\in\mathcal{U}}  \lambda(u) - \sup_{u\in\mathcal{U}} \lambda(u) \sum_{j=q}^q g_j /q$ for some $g\neq \mathrm{id}$. Because this limit is strictly positive, \[ \ev \varphi_{1-\alpha}(\hat{\delta}-\delta_0,\G ) \geq \prob\bigl( T(\hat{\delta}-\delta_0) > T^{(|\G |-1)}(\hat{\delta}-\theta_0,\G ) \bigr) \to 1,\] as required.
\end{proof}

\begin{proof}[Proof of Theorem \ref{th:size2}]
I can work with $X_{n,[h]} = \sqrt{n}(\hat{\delta}_{[h]} - \delta_0 1_{\min\{q_1,q_0\}})$ instead of $\hat{\delta}_{[h]} - \delta_0 1_{\min\{q_1,q_0\}}$ because $x\mapsto p(x, \G)$ is scale invariant. In the following I make repeated use of the fact that the map $x\mapsto (x_{[h]})_{h\in\mathcal{H}}$, the map $x_{[h]}\mapsto (T(gx_{[h]}))_{g\in\G}$, and their composition are continuous.

Suppose $q_1\leq q_0$. The case $q_1 > q_0$ requires only notational changes.  The components of $X_{n, [h]}$ are of the form $\sqrt{n}(\hat{\delta}_{j,h(j)} - \delta_0) = \sqrt{n}(\hat{\delta}_{j,h(j)} - \delta)$ under the null hypothesis. By Assumption~\ref{as:weakc2}, these components converge in distribution to $X_{[h]} := (X_{1,h(1)}, \dots, X_{q_1,h(q_1)})$ jointly in $h$. The same arguments as in the proof of Theorem~\ref{th:size} imply that $T(g X_{n,[h]})$ converges in distribution, jointly in $h$ and $g$, to $T(g X_{[h]})$. For the same reasons as in the proof of Theorem~\ref{th:size}, for a given $h$, $(T(gX_{[h]}))_{g\in\G \setminus \id }$ has no ties $T(X_{[h]})$ with probability $1$, provided Assumption~\ref{as:weakc2} holds.

 Consider \[ |\G| \sum_{h\in\mathcal{H}} p(X_{n, [h]}, \G) = \sum_{h\in\mathcal{H}}\sum_{g\in\G} 1\{ T(gX_{n, [h]}) \geq  T(X_{n, [h]})\}. \]This function jumps discretely if, for some $h$ and $g$, $T(gX_{n, [h]}) =  T(X_{n, [h]})$. The continuous mapping theorem applies to this function if the probability of hitting these jumps is zero, i.e., $\prob(T(gX_{n, [h]}) =  T(X_{n, [h]}) \text{ for some } g\in\G, h\in\mathcal{H}) = 0$. The union bound implies that this probability cannot exceed $\sum_{h\in\mathcal{H}}\sum_{g\in\G}\prob(T(gX_{n, [h]}) =  T(X_{n, [h]})) = 0$ because $(T(gX_{[h]}))_{g\in\G}$ has no ties almost surely. Conclude that the preceding display converges in distribution to  $\sum_{h\in\mathcal{H}}\sum_{g\in\G} 1\{ T(gX_{[h]}) \geq  T(X_{[h]})\}$ and therefore \[ \prob\Biggl( \frac{2}{H} \sum_{h\in\mathcal{H}} p(X_{n,[h]}, \G) \leq \alpha \Biggr) \to \prob\Biggl( \frac{2}{H} \sum_{h\in\mathcal{H}} p(X_{[h]}, \G) \leq \alpha \Biggr) \] if $\alpha$ is a continuity point of the right-hand side. Because $|\G| \sum_{h\in\mathcal{H}} p(X_{[h]}, \G)$ is integer-valued, non-integer values of $\alpha H |\G|/2$ are continuity points. 
 
For integer $\alpha H |\G|/2$, find an $\varepsilon > 0$ such that $(\alpha + \varepsilon) H |\G|/2$ is not an integer but $\alpha + \varepsilon \leq 1$. Monotonicity and weak convergence imply \[ \limsup_{n\to\infty}\prob\Biggl( \frac{2}{H} \sum_{h\in\mathcal{H}} p(X_{n,[h]}, \G) \leq \alpha \Biggr) \leq \prob\Biggl( \frac{2}{H} \sum_{h\in\mathcal{H}} p(X_{[h]}, \G) \leq \alpha + \varepsilon \Biggr). \] By the \citet{ruschendorf1982} inequality, the right-hand side cannot exceed $\alpha + \varepsilon$. Now let $\varepsilon \searrow 0$ to obtain the desired result.
\end{proof}

\begin{proof}[Proof of Theorem \ref{th:power2}]
As in the proof of Theorem \ref{th:size2}, let $X_{n,[h]} = \sqrt{n}(\hat{\delta}_{[h]} - \delta_0 1_{\min\{q_1,q_0\}})$ and $q_1\leq q_0$ without loss of generality. Consider fixed alternatives $\delta = \delta_0 + \lambda$. The components of $\hat{\delta}_{[h]} - \delta_{0}1_{q_1}$ are of the form \[\hat{\delta}_{j,h(j)} - \delta_0  = \sqrt{n}(\hat{\delta}_{j,k} - \delta)/\sqrt{n} +\lambda \leadsto \lambda \] by uniform continuity. Deduce that for every $g$ and $h$, $T(X_{n,[h]}/\sqrt{n})-T(g X_{n,[h]}/\sqrt{n})$ converges in probability to $T(\lambda_{[h]})-T(g \lambda_{[h]})$. For $g\neq \mathrm{id}$, this limit equals \[\sup_{u\in\mathcal{U}} \lambda(u) - \sup_{u\in\mathcal{U}} \lambda(u) \sum_{j=q}^q g_j/q > 0.\] Zero is therefore a continuity point of the (degenerate) limiting distribution of  $T(X_{n,[h]}/\sqrt{n})-T(g X_{n,[h]}/\sqrt{n})$, which implies \[ \prob\Bigl(T(g X_{n,[h]}/\sqrt{n}) \geq T(X_{n,[h]}/\sqrt{n})\Bigr) \to 0\] and $1\{ T(g X_{n,[h]}/\sqrt{n}) \geq T(X_{n,[h]}/\sqrt{n}) \} \to 0$ for every $g\neq \mathrm{id}$ and $h$. Conclude that \[ \frac{1}{H} \sum_{h\in\mathcal{H}} p(X_{n, [h]}, \G) = \frac{1}{|\G|H}\sum_{h\in\mathcal{H}}\sum_{g\in\G} 1\{ T(gX_{n, [h]}) \geq  T(X_{n, [h]})\} \leadsto \frac{1}{|\G|} \] and therefore \[ \prob\Biggl( \frac{2}{H}\sum_{h\in\mathcal{H}} p(\hat{\delta}_{[h]} - \delta_0 1_{q_1}, \G) \leq \alpha \Biggr) \to 1\{ 2 \leq \alpha |\G| \} \] as long as $\alpha|\G| \neq 2$ to guarantee that convergence occurs at a continuity point.

Now consider local alternatives $u \mapsto \delta(u) = \delta_0(u) + c\lambda(u)/\sqrt{n}$ with $c$ constant. As in the proof of Theorem~\ref{th:size}, continuity of the maps $x \mapsto x_{[h]}$ and $x_{[h]} \mapsto (T(gx_{[h]}))_{g\in\G }$ implies $(T(g(X_{n,[h]} + c\lambda 1_{q_1})))_{g\in\G }\leadsto (T(g(X_{[h]} + c\lambda 1_{q_1})))_{g\in\G }$ jointly in $h\in\mathcal{H}$. For a given $h$, $(T(g(X_{[h]} + c\lambda 1_{q_1})))_{g\in\G }$ again has no ties with probability $1$. As before, deduce \[ \prob\Biggl( \frac{2}{H} \sum_{h\in\mathcal{H}} p(X_{n,[h]}  + c\lambda 1_{q_1}, \G) \leq \alpha \Biggr) \to \prob\Biggl( \frac{2}{H} \sum_{h\in\mathcal{H}} p(X_{[h]}+ c\lambda 1_{q_1}, \G) \leq \alpha \Biggr) \] if $\alpha$ is a continuity point of the right-hand side. Because $|\G| \sum_{h\in\mathcal{H}} p(X_{[h]}+ c\lambda 1_{q_1}, \G)$ is integer-valued, non-integer values of $\alpha H |\G|/2$ are continuity points. For integer $\alpha H |\G|/2$, find an $\varepsilon > 0$ such that $(\alpha - \varepsilon) H |\G|/2$ is not an integer but $\alpha - \varepsilon > 0$. \[ \liminf_{n\to\infty}\prob\Biggl( \frac{2}{H} \sum_{h\in\mathcal{H}} p(X_{n,[h]}  + c\lambda 1_{q_1}, \G) \leq \alpha \Biggr) \geq \prob\Biggl( \frac{2}{H} \sum_{h\in\mathcal{H}} p(X_{[h]}+ c\lambda 1_{q_1}, \G) \leq \alpha - \varepsilon \Biggr) \] by monotonicity. Let $\varepsilon \searrow 0$ to see that the limit inferior is bounded below by \[ \prob\Biggl( \frac{2}{H} \sum_{h\in\mathcal{H}} p(X_{[h]}+ c\lambda 1_{q_1}) < \alpha \Biggr). \] The same bound holds trivially for non-integer $\alpha H |\G|/2$.

For the analysis as $c \to \infty$, consider \[ \frac{2}{H} \sum_{h\in\mathcal{H}} p(X_{[h]}+ c\lambda 1_{q_1}) = \frac{2}{|\G|H}\sum_{h\in\mathcal{H}}\sum_{g\in\G} 1\Biggr\{ \frac{T(g(X_{[h]}+ c\lambda 1_{q_1}))}{T(c\lambda 1_{q_1})} \geq  \frac{T(X_{[h]}+ c\lambda 1_{q_1})}{T(c\lambda 1_{q_1})}\Biggr\}. \] For $g = \mathrm{id}$, the indicator function in the preceding display equals $q$. Consider $g\neq \mathrm{id}$. Because $T(c\lambda 1_{q_1}) = c T(\lambda 1_{q_1})  > 0$ and $T(gX_{[h]})/T(c\lambda 1_{q_1})\to 0$ almost surely for every $g\in\G$ as $c\to\infty$, it follows from the subadditivity of suprema that $T(g(X_{[h]} + c\lambda 1_{q_1}))/T(c\lambda 1_{q_1})\to T(g\lambda 1_{q_1})/T(\lambda 1_{q_1})$ almost surely and therefore $(T(X_{[h]} + c\lambda 1_{q_1}) - T(g(X_{[h]} + c\lambda 1_{q_1})))/T(c\lambda 1_{q_1}) \to 1 - (T(g\lambda 1_{q_1})/T(\lambda 1_{q_1}))$ almost surely. That last limit is a strictly positive constant for every $g\neq \mathrm{id}$ and there is one $\mathrm{id}$ for every $h$. Conclude from the continuous mapping theorem that the preceding display converges almost surely to $2/|\G|$ as $c\to\infty$. If $\alpha|\G| \neq 2$, it follows that \[ \lim_{c \to \infty} \prob\Biggl( \frac{2}{H} \sum_{h\in\mathcal{H}} p(X_{[h]}+ c\lambda 1_q) < \alpha \Biggr) = 1\{ 2 < \alpha |\G| \},  \] as required.
\end{proof}

\begin{proof}[Proof of Proposition \ref{pr:randomrueschendorf}]
	If $\mathcal{I}$ is fixed, then the proof of Theorems \ref{th:size2} and \ref{th:power2} goes through without any changes. For random $\mathcal{I}$, work conditional on $\mathcal{I}$ to see that Theorem \ref{th:size2} implies \[ \limsup_{n\to\infty}\prob\Biggl( \frac{2}{|\mathcal{I}|}\sum_{h\in\mathcal{I}} p(\hat{\delta}_{[h]} - \delta_0 1_{\min\{q_1,q_0\}}, \G) \leq \alpha \mid \mathcal{I} \Biggr)\leq \alpha\] almost surely. Apply expectations to conclude from the (reverse) Fatou lemma that 
	\begin{align*}
		&\limsup_{n\to\infty}\prob\Biggl( \frac{2}{|\mathcal{I}|}\sum_{h\in\mathcal{I}} p(\hat{\delta}_{[h]} - \delta_0 1_{\min\{q_1,q_0\}}, \G) \leq \alpha\Biggr) \\ 
		&\qquad \leq\ev \limsup_{n\to\infty}\prob\Biggl( \frac{2}{|\mathcal{I}|}\sum_{h\in\mathcal{I}} p(\hat{\delta}_{[h]} - \delta_0 1_{\min\{q_1,q_0\}}, \G) \leq \alpha \mid \mathcal{I} \Biggr)\leq \alpha
	\end{align*}
as needed. Similarly, Fatou's lemma implies
	\begin{align*}
		&\liminf_{n\to\infty}\prob\Biggl( \frac{2}{|\mathcal{I}|}\sum_{h\in\mathcal{I}} p(\hat{\delta}_{[h]} - \delta_0 1_{\min\{q_1,q_0\}}, \G) \leq \alpha\Biggr) \\ 
		&\qquad \geq\ev \liminf_{n\to\infty}\prob\Biggl( \frac{2}{|\mathcal{I}|}\sum_{h\in\mathcal{I}} p(\hat{\delta}_{[h]} - \delta_0 1_{\min\{q_1,q_0\}}, \G) \leq \alpha \mid \mathcal{I} \Biggr).
	\end{align*}
	Now apply the first part of Theorem \ref{th:power2} for a given $\mathcal{I}$ to get the result for fixed alternatives. For local alternatives, the proof of Theorem \ref{th:power2} implies 
	\begin{align*}
	&\liminf_{n\to\infty}\prob\Biggl( \frac{2}{|\mathcal{I}|}\sum_{h\in\mathcal{I}} p(\hat{\delta}_{[h]} - \delta_0 1_{\min\{q_1,q_0\}}, \G) \leq \alpha \mid \mathcal{I} \Biggr)\\	
	&\qquad  \geq \prob\Biggl( \frac{2}{|\mathcal{I}|} \sum_{h\in\mathcal{I}} p(X_{[h]}+ c\lambda 1_{q_1}) < \alpha \mid \mathcal{I} \Biggr) \to 1
	\end{align*}
	almost surely as $c \to \infty$, as required.
\end{proof}

\begin{proof}[Proof of Proposition \ref{pr:stochastic}]
	Limits are as $m\to\infty$ unless noted otherwise. Consider a process $X_n$ possibly depending on $n$ and recall that  $T(X_n) > T^{1-\alpha}(X_n, \G_m)$ if and only if $\hat{p}_m := p(X_n,\G_m) \leq \alpha$. Let $p := p(X_n,\G)$ and notice that $\ev (\hat{p}\mid X_n) = p$. For almost every realization of $X_n$,  $\hat{p}_m$ is an average of bounded iid random variables that satisfies $\prob(|\hat{p}_m - p| > \varepsilon\mid X_n) \to 0$ almost surely for every $\varepsilon > 0$. Conclude from dominated convergence that this convergence also holds unconditionally and therefore $\hat{p}_m\leadsto p$. Because $p$ can only vary at the points $j/|\G|$, $1\leq j\leq |G|$, $\prob(\hat{p}_m \leq \alpha) \to \prob(p \leq \alpha)$ as long as $\alpha \neq j/|\G|$. If $\alpha$ equals $j/|\G|$ for some $j$, use $0 < \varepsilon < 1/|\G|$ and monotonicity to see that $\prob(\hat{p}_m \leq \alpha - \varepsilon) \leq \prob(\hat{p}_m \leq \alpha) \leq \prob(\hat{p}_m \leq \alpha + \varepsilon)$ must satisfy \[ \prob(p\leq \alpha - \varepsilon) \leq \liminf_{m\to\infty}\prob(\hat{p}_m \leq \alpha)\leq \limsup_{m\to\infty}\prob(\hat{p}_m \leq \alpha) \leq \prob(p\leq \alpha + \varepsilon). \] Let $\varepsilon \searrow 0$ to see that the extreme right-hand side can be decreased to $\prob(p\leq \alpha)$.
	
For Theorem \ref{th:size}, apply this result to obtain \[ \limsup_{m\to\infty}  \prob\bigl(T(X_n) > T^{1-\alpha}(X_n, \G_m)\bigr) \leq \prob\bigl( p(X_n,\G) \leq \alpha \bigr) = \ev \varphi_\alpha(X_n,\G). \] Now apply limits as $n\to\infty$. 

For Theorem \ref{th:size2}, consider stochastic processes $X_{n,h}$ indexed by $h$ and $n$. The continuous mapping theorem implies $2\sum_{h\in\mathcal{H}} p(X_{n, h},\G_m)/H \pto 2\sum_{h\in\mathcal{H}} p(X_{n, h},\G)/H$ and therefore $2\sum_{h\in\mathcal{H}} p(X_{n, h},\G_m)/H \leadsto 2\sum_{h\in\mathcal{H}} p(X_{n, h},\G)/H$. Using the same argument as before gives
\begin{align*}
\limsup_{m\to\infty}\prob\Biggl( \frac{2}{H}\sum_{h\in\mathcal{H}} p(X_{n, h}, \G_m) \leq \alpha \Biggr) \leq \prob\Biggl( \frac{2}{H}\sum_{h\in\mathcal{H}} p(X_{n, h}, \G) \leq \alpha \Biggr)
\end{align*}
	
For Theorem \ref{th:power}, if $\alpha > 1/2^q$, there is a $\varepsilon > 0$ such that $\alpha - \varepsilon > 1/2^q$.  Then \[ \liminf_{m\to\infty}  \prob\bigl(T(X_n) > T^{1-\alpha}(X_n, \G_m)\bigr) \geq \prob\bigl( p(X_n,\G) \leq \alpha - \varepsilon \bigr) = \ev \varphi_{\alpha-\varepsilon}(X_n,\G) \] and Theorem \ref{th:power} applies directly to the extreme right-hand side.

For Theorem \ref{th:power2}, there is a $\varepsilon > 0$ such that $\alpha - \varepsilon > 1/2^{q-1}$. Then \[
\liminf_{m\to\infty}\prob\Biggl( \frac{2}{H}\sum_{h\in\mathcal{H}} p(X_{n, h}, \G_m) \leq \alpha \Biggr) \geq \prob\Biggl( \frac{2}{H}\sum_{h\in\mathcal{H}} p(X_{n, h}, \G) \leq \alpha -\varepsilon \Biggr)\] and Theorem \ref{th:power2} can be applied to the extreme right-hand side.
\end{proof}

\end{document}